\documentclass[12pt]{article}
\usepackage{epsfig,amsmath,amssymb,amsfonts,amstext,amsthm,mathrsfs}
\usepackage{latexsym,graphics,epsf,epsfig,psfrag}
\usepackage{cite}

\newtheorem{definition}{\textbf{Definition}}
\newtheorem{corollary}{\textbf{Corollary}}
\newtheorem{lemma}{\textbf{Lemma}}
\newtheorem{theorem}{\textbf{Theorem}}

\newtheorem{remark}{\textbf{Remark}}

\newtheorem{thm}{Theorem}

\newtheorem{defn}{Definition}

\newcommand{\nn}{\nonumber}

\newcommand{\cC}{\mathcal{C}}

\newcommand{\hw}{\hat{w}}
\newcommand{\hl}{\hat{l}}

\DeclareMathAlphabet{\matheuf}{U}{euf}{m}{n}

\begin{document}

\vspace*{-2cm}

\begin{center}
  \baselineskip 1.3ex {\Large \bf An Information Theoretic Approach to Secret Sharing
  \footnote{The material in this paper was presented
in part at the IEEE International Symposium on Information Theory (ISIT), Saint Petersburg, Russia, August
2011~\cite{Liang11isit} and at the IEEE 14th Workshop on Signal Processing Advances in Wireless Communications (SPAWC), Darmstadt, Spain, June 2013~\cite{zou2013}.} \footnote{The work of S. Zou and Y. Liang was supported by a National Science Foundation
CAREER Award under Grant CCF-10-26565 and by the National Science Foundation under Grants CCF-10-26566 and CNS-11-16932. The work of L. Lai was supported by a National Science Foundation CAREER Award under Grant CCF-13-18980 and the National Science Foundation under Grant CNS-13-21223. The work of S. Shamai (Shitz) was supported by the Israel Science Foundation (ISF), and the European Commission in the framework of the Network of Excellence in Wireless COMmunications NEWCOM$\#$.}
}
\\
 \vspace{0.15in} Shaofeng Zou, Yingbin Liang, Lifeng Lai, and Shlomo Shamai (Shitz)
\footnote{Shaofeng Zou and Yingbin Liang are with the Department of Electrical
Engineering and Computer Science, Syracuse University, Syracuse, NY 13244 USA (email: \{szou02,yliang06\}@syr.edu).
Lifeng Lai is with the Department of Electrical and Computer Engineering, Worcester Polytechnic Institute, Worcester, MA 01609 USA (email: llai@wpi.edu).
Shlomo Shamai (Shitz) is
with the Department of Electrical Engineering, Technion-Israel
Institute of Technology, Technion City, Haifa 32000 Israel (email: sshlomo@ee.technion.ac.il).}
\end{center}

\begin{abstract}
\baselineskip 3.5ex
A novel information theoretic approach is proposed to solve the secret sharing problem, in which a dealer distributes one or multiple secrets among a set of participants in such a manner that for each secret only qualified sets of users can recover this secret by pooling their shares together while non-qualified sets of users obtain no information about the secret even if they pool their shares together. While existing secret sharing systems (implicitly) assume that communications between the dealer and participants are noiseless, this paper takes a more practical assumption that the dealer delivers shares to the participants via a noisy broadcast channel. Thus, in contrast to the existing solutions that are mainly based on number theoretic tools, an information theoretic approach is proposed, which exploits the channel randomness during delivery of shares as additional resources to achieve secret sharing requirements. In this way, secret sharing problems can be reformulated as equivalent secure communication problems via wiretap channel models, and can hence be solved by employing powerful information theoretic security techniques. This approach is first developed for the classic secret sharing problem, in which only one secret is to be shared. This classic problem is shown to be equivalent to a communication problem over a compound wiretap channel. Thus, the lower and upper bounds on the secrecy capacity of the compound channel provide the corresponding bounds on the secret sharing rate, and the secrecy scheme designed for the compound channel provides the secret sharing schemes. The power of the approach is further demonstrated by a more general layered multi-secret sharing problem, which is shown to be equivalent to the degraded broadcast multiple-input multiple-output (MIMO) channel with layered decoding and  secrecy constraints. The secrecy capacity region for the degraded MIMO broadcast channel is characterized, which provides the secret sharing capacity region. Furthermore, the secure encoding scheme that achieves the secrecy capacity region provides an information theoretic scheme for sharing the secrets.
\end{abstract}

\baselineskip 3.ex
\section{Introduction}

In the classic secret sharing problem, a dealer intends to distribute a secret among a set of participants  such that only qualified sets of participants can correctly recover the secret by pooling their shares together, while the non-qualified set of participants obtain no information about the secret even if they pool their shares together. There are rich applications of secret sharing including construction of protocols and algorithms for secure multiparty computations \cite{Yao1982,Yao1986}, Byzatine agreement \cite{Rabin1983}, threshold cryptography \cite{Desmedt1992}, access control \cite{Naor1996}, attribute-based encryption \cite{Goyal2006}, and generalized oblivious transfer \cite{Shankar2008}. The existing solutions for the secret sharing problems are mainly based on the number theoretic tools, in which contents of the shares that the dealer delivers to the participants are specially designed in order to guarantee the secret sharing requirements. While such approaches work well for simple secret sharing problems, they are not readily extendable to more complicated problems, in which qualified and non-qualified sets become more complicated, and/or multiple secrets are simultaneously shared.


While in existing secret sharing systems, it is implicitly assumed that communications between the dealer and participants are noiseless, in this paper we take a more practical assumption that the dealer delivers shares to the participants via a noisy broadcast channel. Thus, we propose a novel information theoretic approach to solving secret sharing problems, which exploits the channel randomness during delivery of shares from dealers to participants as additional resources to achieve secret sharing requirements. In this way, secret sharing problems can be equivalently reformulated into secure communication problems via wiretap channel models studied in information theory \cite{Wyner1975,Csiszar1978} (see \cite{Liang2009} and \cite{Bloch2011} for more references of these studies). More importantly, such an approach is general enough to incorporate complex secret sharing requirements with multiple secrets and various structures of qualified and non-qualified sets into wiretap models. These secure communication problems can then be solved by employing powerful information theoretic security techniques, which thus enable the design of secret sharing strategies. Furthermore, such an approach also allows one to characterize the secret sharing capacity region based on the information theoretic characterization of the secrecy capacity region of physical layer wiretap models.

We first illustrate the basic idea of our approach using the classic secret sharing problem, in which the dealer wishes to distribute one secret to qualified sets of participants specified by an arbitrary access structure. We propose to achieve the secret sharing property via broadcasting the secret message from the dealer to all participants. To design a secret sharing scheme, we construct an equivalent compound wiretap channel~\cite{Liang:2009} by creating one virtual legitimate receiver for each qualified set (with the receiver's output including channel outputs from all participants in the qualified set) and creating one virtual eavesdropper for each non-qualified set (with the eavesdropper's output including channel outputs from all participants in the non-qualified set). Thus, a secure communication scheme for the equivalent compound wiretap channel guarantees that the qualified sets of participants can decode the secret while the non-qualified sets of participants obtain a negligible amount of information about the secret. By applying the results for the compound wiretap channel in \cite{Liang:2009}, we obtain a secret sharing scheme and characterize lower and upper bounds on the secret sharing rate. We also characterize the secret sharing capacity for some secret sharing scenarios.

We then demonstrate the power of our approach via a more complicated multi-secret sharing problem. Although we here solve an example problem, our goal is to demonstrate that this approach can be applied to more general secret sharing problems, not limited to the one that we present in this paper. We consider the following multi-secret sharing problem (see Fig.~\ref{secretmodel}), in which multiple secrets are intended for different corresponding qualified sets to recover in a layered fashion. More specifically, a dealer equipped with multiple antennas wishes to distribute $K$ secrets to $K$ participants by broadcasting via multiple antennas over a wireless channel. It is required that participant 1 recover the first secret, and as one more participant joins the group to share its output, one more secret should be recovered by the group, and this new secret should be kept secure from the smaller groups. Hence, if the first $k$ of the $K$ participants share their channel outputs, they can recover the first $k$ secrets, while all the remaining secrets are kept confidential to the group of the first $k$ participants.

The above problem involves sharing multiple secrets in a layered fashion, and can be very challenging to solve using the traditional number theoretic tools. To solve this problem, following the approach developed above, we design a virtual receiver for each qualified set (i.e., sharing group) of receivers, and thus convert this secret sharing problem to a communication problem over the degraded Gaussian multiple-input multiple-output (MIMO) broadcast channel with layered decoding and secrecy constraints. The requirements of the secret sharing problem is exactly mapped into the layered decoding and secrecy requirements for the communication problem. More specifically, in the communication model (see Fig.~\ref{model} for an illustration), the transmitter wishes to transmit $K$ messages to $K$ (virtual) receivers. The channel outputs at receivers naturally satisfy the degradedness condition, i.e., from receiver $K$ to receiver $1$, the quality of their channels gets worse gradually. It is required that receiver $k$ decodes one more message than receiver $k-1$ for $k=2,\ldots,K$, and this additional message should be kept secure from all receivers with worse channel outputs, i.e., with lower indices. Design of secure communication schemes and characterization of the secrecy capacity region for such a channel model readily provide secret sharing schemes and the corresponding secret sharing rate regions.

Characterizing the secrecy capacity region of the degraded Gaussian MIMO broadcast channel with layered decoding and secrecy constraints is a challenging information theoretic problem. Previously, there have been a number of broadcast models with layered decoding and secrecy requirements proposed and studied. In particular, \cite{Liu2010} studied a model (model 1) with two legitimate receivers and one eavesdropper. It is required that one message be decoded at both receivers and kept secure from the eavesdropper, and that the second message be decoded at one receiver and kept secure from the other receiver and the eavesdropper. \cite{Liu2010} studied one more model (model 2), in which the second message does not need to be kept secure from the other receiver. Both models were further generalized in \cite{ekrem2012} in that each receiver and the eavesdropper in the above model was replaced by a group of nodes. In \cite{Liu2010,ekrem2012}, the secrecy capacity region was established for the MIMO Gaussian channels.


The model that we study in this paper generalizes model 1 in \cite{Liu2010} to $K$ receivers. We characterize the secrecy capacity region for the discrete memoryless channel and the single-input-single-output (SISO) and MIMO Gaussian channels with layered decoding and secrecy constraints. Towards this end, the challenges lie in both the achievability and converse proofs due to the layered secrecy constraints on more than two receivers. Our achievable scheme is based on the stochastic encoding (i.e., binning) and heterogeneous superposition. The major challenge of achievability arises in the analysis of leakage rates, which is much more involved than the cases with two secure messages. Our contribution here lies in novel generalization of the analysis of the leakage rate provided in \cite{Gamal2012} for one secure message to multiple secure messages. On the other hand, due to the layered secrecy constraints, outer bounds on the secrecy rates should be developed in certain recursive structures for three or more consecutive layers of receivers. Consequently, techniques used in \cite{Liu2010,ekrem2012} for two layers cannot be readily applied here, although some properties on matrix manipulations are useful in our proof for the MIMO channel. Our main technical development in the converse proof lies in the construction of a series of covariance matrices representing input resources for layered messages such that the secrecy rates can be upper bounded as the desired recursive forms in terms of these covariance matrices.



Due to the equivalence of the multi-secret sharing problem via the multiple-input-single-output (MISO) broadcast channel and the secure communication problem over the degraded MIMO broadcast channel, we hence establish the secret sharing capacity region. Furthermore, the secure encoding scheme that achieves the secret capacity region provides an information theoretic scheme for sharing multiple secrets.

The rest of the paper is organized as follows. In Section \ref{sec:sec1}, we introduce the classical secret sharing problem and its connection to an\cite{ekrem2012}\cite{ekrem2011secrecy}  information theoretic secrecy model. We then use this secret sharing problem to illustrate our information theoretic approach for secret sharing. As further demonstration of our approach, in Sections \ref{sec:model} and \ref{sec:app}, we study an information theoretic model of the degraded broadcast channel with layered decoding and secrecy constraints, and apply the results obtained to study a layered multi-secret sharing problem. Finally, in Section \ref{sec:con}, we conclude our paper with further remarks.

\section{Secret Sharing and Information Theoretic Secrecy}\label{sec:sec1}
In this section, we first introduce the secret sharing problem, and then connect this problem to a model studied in information theoretic secrecy.

\subsection{Model of Secret Sharing}
We consider the following secret sharing problem. Suppose the system consists of
a dealer and a set of participants $\mathcal{P}=\{1,2,\cdots, K\}$.
The dealer has a secret $W$ (taken from a set $\mathcal{W}$) for the $K$ participants to share. We define an access structure $\mathcal{A}$, which contains all subsets of $\mathcal{P}$ that are required to recover the secret. Each set $A\in\mathcal{A}$ is called a qualified set. We assume that the access structure considered in this paper is monotone~\cite{Blundo:TIT:95}, that is if $A\in \mathcal{A}$ and $A\subseteq A_1$, then $A_1\in\mathcal{A}$. For the secret sharing scheme, we require that if the users in any qualified set $A\in\mathcal{A}$ gather their observations together, then they can recover the secret with a negligible error probability. We define a non-access structure $\mathcal{B}$ such that for any set $B\not\in \mathcal{B}$, we require that even if users in the set $B$ gather their observations together, they obtain negligible information about the secret message. In many applications, $\mathcal B=\mathcal A^C$.

In the existing secret sharing schemes, the communications between the dealer and participants are assumed to be noiseless as the classic secret sharing problem does not involve channel. In this paper,
we assume that the dealer and the participants are connected by a noisy broadcast channel, as shown in Figure~\ref{fig:channel}. If the dealer transmits $X^{n}$, participant $k$ receives $Y_k^n$, and the relationship among the input and outputs is characterized by the transition probability distribution
\begin{eqnarray}
P_{Y_1^n\cdots Y_K^n|X^n}(y_1^n\cdots y_K^n|x^n)=\prod\limits_{k=1}^K\prod\limits_{i=1}^n P_{Y_k|X}(y_k(i)|x(i)),\label{eq:channel}
\end{eqnarray}
where $x(i)$ is taken from a finite set $\mathcal{X}$, and $y_k(i)$ is taken from a finite set $\mathcal{Y}_k$ for $k=1,\ldots,K$.

\begin{figure}[thb]
\centering
\includegraphics[width=0.5 \textwidth]{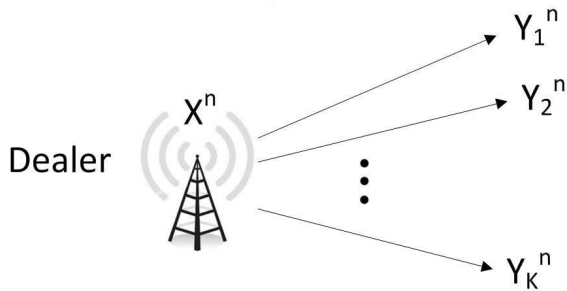}
\caption{The noisy broadcast channel.} \label{fig:channel}
\end{figure}

\begin{defn}
A $(2^{nR},n)$ code for secret sharing over the noisy broadcast channel consists of the following:
\begin{enumerate}
\item a secret set: $\mathcal{W}=\{1,2,\cdots,2^{nR}\}$ with the secret $W$ uniformly distributed over $\mathcal{W}$;
\item an encoder $f$: $\mathcal{W}\rightarrow \mathcal{X}^n$ mapping each secret message $w\in\mathcal{W}$ to a codeword $x^n\in\mathcal{X}^n$;
\item a decoder $g_{A}$ for each qualified set $A\in \mathcal{A}$: $\{\mathcal{Y}_k^n:k\in A\}\rightarrow \mathcal{W}^{A}$.
\end{enumerate}
\end{defn}
The average block error probability for set $A$ is
\begin{eqnarray}
P_{e,A}^n=\frac{1}{2^{nR}}\sum\limits_{w=1}^{2^{nR}}\text{Pr}\{\hat{w}^{A}\neq w|w\text{ was sent}\}.
\end{eqnarray}

\begin{defn}
A secret sharing rate $R$ is said to be achievable if there exists a sequence of $(2^{nR},n)$ codes such that for any $\epsilon>0$ the following two conditions are satisfied:
 \begin{enumerate}
 \item $\forall A\in \mathcal{A}$, we have
 \begin{eqnarray}
 P_{e,A}^n\leq \epsilon;\label{eq:error}
 \end{eqnarray}
 \item $\forall A\in \mathcal{B}$, we have
 \begin{eqnarray}
 \frac{1}{n}I(W;\{Y^n_{k}:k\in B\})\leq \epsilon.\label{eq:leak}
 \end{eqnarray}
 \end{enumerate}
\end{defn}
The above condition~\eqref{eq:error} requires that the decoding error probability at any qualified set should be small, while condition~\eqref{eq:leak} requires that any non-qualified set gains negligible amount of information about the secret even if users in this set share their observations together.

The {\em secret sharing capacity} is defined to be the maximal achievable secret sharing rate.


\subsection{Connection to Wiretap Channels}
In the remaining part of this section, we study the above secret sharing problem with non-access structure $\mathcal B=\mathcal A^C$.
Our main idea is to connect this problem to a communication problem over an equivalent compound wiretap channel ~\cite{Liang:2009} (i.e., the Wyner's wiretap channel model with multiple legitimate receivers and multiple eavesdroppers). Thus, secure coding design for the equivalent compound wiretap channel can be applied to achieve secret sharing. 

More specifically, for each set $A\in \mathcal{A}$, we construct a virtual legitimate receiver $V_{A}$ such that the observation at the virtual receiver is $Y_{V_{A}}=\{Y_k:k\in A\}$.
Clearly, this receiver is not an actual node, but an identity representing that the users in set $A$ share their outputs.
For an access structure $\mathcal{A}$, we will construct $|\mathcal{A}|$ virtual legitimate receivers, representing that in each of these sharing scenarios, the secret message is required to be recovered.
In addition, for each set $B\in \mathcal{B}$, we construct a virtual wiretapper $V_{E,B}$ with the observation $Y_{V_{E,B}}=\{Y_k:k\in B\}$.
These virtual eavesdroppers represent that in each of these sharing scenarios, the message should be kept secure. We note that these virtual eavesdroppers are also not actual devices, but the identities representing that nodes in non-qualified sets share their outputs.


We note that if $A_1\subset A_2$, then $X\rightarrow Y_{V_{A_2}}\rightarrow Y_{V_{A_1}}$. In this case, if we design a code such that the users in the qualified set $A_1$ can decode the secret, then the users in the qualified set $A_2$ can also decode the secret. Hence, it is not necessary to construct a virtual legitimate receiver for the qualified set $A_2$. In this way, the number of virtual legitimate receivers can be reduced in the constructed equivalent compound wiretap channel. Similarly, if $B_1\subset B_2$, then it is not necessary to construct a virtual wiretapper for the set $B_1$, since if the message is kept secure from the set $B_2$, then it is also kept secret from the set $B_1$. Hence, we can also reduce the number of constructed virtual wiretappers.

Figure~\ref{fig:compoundwc} illustrates an equivalent compound wiretap channel for a secret sharing system with four participants. Any three participants are required to recover the secret, and hence the access structure $\mathcal{A}$ includes four qualified sets $\{1,2,3\}, \{1,2,4\}, \{1,3,4\}, \{2,3,4\}$. We note that the set $\{1,2,3,4\}$ is not included as a virtual legitimate receiver due to the reason mentioned above. Furthermore, any two participants should not recover the secret, and hence six virtual eavesdroppers are created corresponding to these non-qualified sets $\{1,2\}, \{1,3\}, \{1,4\}, \{2,3\}, \{2,4\}, \{3,4\}$. 

\begin{figure}[thb]
\centering
\includegraphics[width=0.5 \textwidth]{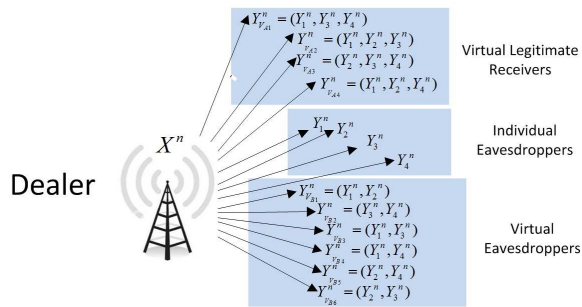}
\caption{An equivalent compound wiretap channel for a secret sharing with four participants.} \label{fig:compoundwc}
\end{figure}

Secure coding schemes for the compound wiretap channel have been proposed in \cite{Liang:2009}, which guarantee that all legitimate receivers recover the message and none of the wiretappers obtain any information about the message. By applying the approach developed in \cite{Liang:2009} to the equivalent compound wiretap channel that we construct for the secret sharing problem, we obtain corresponding schemes such that each virtual legitimate receiver can decode the message $W$ while each virtual wiretapper has negligible information about the message $W$, which are exactly the conditions required by secret sharing. The following bounds on the secret sharing capacity follows from the bounds on the secrecy capacity in \cite{Liang:2009}.
\begin{thm}\label{thm:general}
The following secret sharing rate is achievable via the broadcast channel described in~\eqref{eq:channel}
\begin{eqnarray}\label{eq:U}
R_N=\max\limits_{P_{UX}}\left[\min\limits_{A}I(U;Y_{V_A})-\max\limits_{B}I(U;Y_{V_{E,B}})\right],
\end{eqnarray}
where $U$ is an auxiliary random variable that satisfies the Markov chain relationship:
\begin{eqnarray}
U\rightarrow X\rightarrow (Y_1,\cdots, Y_K).
\end{eqnarray}
Furthermore, the following rate is an upper-bound on the secret sharing capacity
\begin{eqnarray}\label{eq:general}
R=\min\limits_{A,B}\max\limits_{P_{UX}P_{Y_{V_A}Y_{V_{E,B}}|X}}[I(U;Y_{V_A})-I(U;Y_{V_{E,B}})].
\end{eqnarray}
\end{thm}

Since the secrecy capacity for the general compound wiretap channel is still unknown, the secrecy sharing capacity for an arbitrary access structure $\mathcal{A}$ is not known. However, there are some interesting special cases of the secret sharing problem, which naturally correspond to degraded compound wiretap channels, and hence the secrecy sharing capacity can be characterized for these cases. For example, we consider the scenario, in which the secret is decodable only if all participants share their outputs, and is kept secret from any subset of participants. The secret sharing capacity is given as follows.
\begin{corollary}
Consider the $K$-participant secret sharing problem over the noisy channel  $P_{Y_1,\ldots,Y_K|X}$. Suppose $\mathcal{A}$ contains a single set $A=\{1, \ldots,K\}$. The secret sharing capacity is given by
\begin{eqnarray}
C_s=\max\limits_{P_X}\min_{B}\left[I(X;Y_1,\ldots,Y_K)-I(X;Y_{V_{E,B}})\right].
\end{eqnarray}
where the set $B$ can be any strict subset of $A=\{1,\ldots,K\}$.
\end{corollary}
For the two-participant secret sharing problem, the secret sharing capacity is given by
\begin{eqnarray}\label{eq:ctwo}
C_{two}=\max\limits_{P_X}\left[I(X;Y_1,Y_2)-\max\{I(X;Y_1),I(X;Y_2)\}\right].
\end{eqnarray}


Another example for which we can fully characterize the secret sharing capacity is the secret sharing problem over Gaussian broadcast channel, in which
\begin{eqnarray}
Y_k=X+Z_k\quad\text{for }k=1,\cdots,K.
\end{eqnarray}
where
$Z_k$, for $k=1,\cdots, K$ are independent thermal Gaussian noise variables with  mean zero and variance $N_k$. The dealer has an average power constraint:
\begin{eqnarray}
\frac{1}{n}\mathbb{E}\left\{\sum\limits_{i=1}^nX^2(i)\right\}\leq P.
\end{eqnarray}


For the Gaussian channel example, one can obtain an analytical form of the achievable secret sharing rate for an arbitrary access structure $\mathcal{A}$ by setting $U$ in~\eqref{eq:U} to be Gaussian random variable jointly distributed with $X$. However, the obtained rate takes a complex form and may not provide any insight. A simpler and often-encountered scenario is the so-called $(k,K)$-secret sharing, in which any $k$ or more users can recover the secret by sharing their observations, while any $k-1$ users obtain only negligible information from their joint observations. This secret sharing problem can be reformulated into the secure communication problem over a compound wiretap channel, which consists of one transmitter with single antenna, multiple virtual legitimate receivers with each having $k$ antennas (i.e., each virtual receiver corresponds to a group of $k$ users), and multiple virtual eavesdroppers with each having $k-1$ antennas (i.e., each virtual eavesdropper corresponds to a group of $k-1$ users). For such a single-input multiple-output (SIMO) compound wiretap channel, the secrecy capacity region can be obtained because the same Gaussian input maximizes the secrecy capacity for each pair of legitimate receiver and eavesdropper, and hence achieves the secrecy capacity of the compound channel. Therefore, for $(k,K)$-secret sharing, we obtain the secret sharing capacity given below.
\begin{corollary}
For the $(k,K)$-secret sharing problem, the secret sharing capacity is given by
\begin{eqnarray}\label{eq:ckj}
C_{k,K}=\min_{A_k, A_{k-1}} \log\left(\frac{1+\sum_{l\in A_k}  P/N_l}{1+\sum_{l \in A_{k-1}}  P/N_l}\right).
\end{eqnarray}
where $A_k$ can be any subset of $\{1,\ldots,K\}$ with $k$ indices, and $A_{k-1}$ can be any subset of $\{1,\ldots,K\}$ with $k-1$ indices.

For the special case when the channels to all receivers are symmetric, i.e., $N_k=1$ for all $k=1,\ldots, K$, the secret sharing capacity is given by
\[C_{k,K}=\log\left(\frac{1+kP}{1+(k-1)P}\right).\]
\end{corollary}
\begin{remark}
It is interesting to note that for the special case with symmetric channels, the secret sharing capacity depends only on $k$ but not on $K$ which is the total number of participants.
\end{remark}

\section{Broadcast Channel with Layered Decoding and Secrecy}\label{sec:model}
In order to study a more general secret sharing problem in which simultaneously sharing multiple secrets is required as we introduce in Section~\ref{sec:app}, we need to study a broadcast wiretap model with layered decoding and secrecy. In this section, we first introduce the system model for this channel, and then we provide our characterization of the secrecy capacity region for this channel. These results will then be applied to the layered multiple secrets sharing problem in Section~\ref{sec:app}.


\subsection{System Model}\label{subsec:model}
\begin{figure}[thb]
\begin{center}
\includegraphics[width=5in]{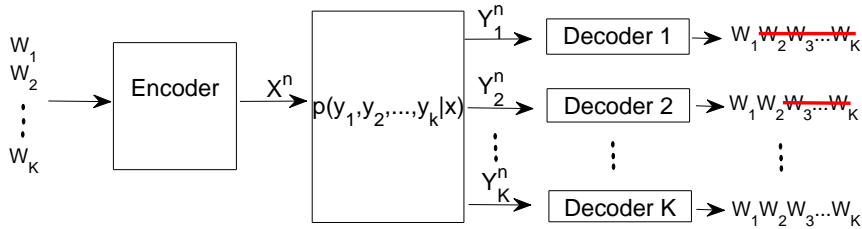}
\caption{The broadcast channel with layered decoding and secrecy}
\label{model}
\end{center}
\end{figure}

In this section, we consider the model of the degraded broadcast channel with layered decoding and secrecy constraints (see Fig.~\ref{model}), in which a transmitter transmits to $K$ receivers. The channel transition probability function is given by $P_{Y_1\cdots Y_K|X}$, in which $X\in\mathcal{X}$ is the channel input and $Y_k\in\mathcal{Y}_k$ is the channel output of receiver $k$ for $k=1,\ldots,K$. It is assumed that the receivers have degraded outputs, i.e., $Y_1,\cdots,Y_K$ satisfy the following Markov chain condition, i.e., the degradedness condition:
\begin{eqnarray}
X\rightarrow Y_K\rightarrow Y_{K-1}\rightarrow\ldots\rightarrow Y_2\rightarrow Y_1.\label{eq:markov}
\end{eqnarray}
Hence, the quality of channels gradually degrades from receiver $K$ to receiver $1$. The transmitter has $K$ messages $W_1,\cdots,W_K$ intended for the $K$ receivers. The system is required to satisfy the following layered decoding and secrecy constraints. For $k=1,\ldots,K$, receiver $k$ needs to decode the messages $W_1,\cdots,W_k$, and to be kept ignorant of messages $W_{k+1},\cdots, W_K$ (see Fig.~\ref{model} for an illustration). We note that the model presented here is well motivated by application of secret sharing problems to be discussed in Section \ref{sec:app}.

A $(2^{nR_1},\cdots,2^{nR_K},n)$ code for the channel consists of
\begin{list}{$\bullet$}{\topsep=0ex \leftmargin=6.5mm \rightmargin=0mm \itemsep=0mm}
\item $K$ message sets: $W_k\in\mathcal{W}_k=\{1,\cdots, 2^{nR_k}\}$ for $k=1,\cdots, K$, which are independent from each other and each message is uniformly distributed over the corresponding message set;
\item An (possibly stochastic) encoder $f^n$: $\mathcal{W}_1\times\cdots\times\mathcal{W}_K\rightarrow \mathcal{X}^n$;
\item $K$ decoders $g_k^n: \mathcal{Y}_k^n\rightarrow (\mathcal{W}_1,\cdots,\mathcal{W}_k)$ for $k=1,\cdots, K$.
\end{list}

Hence, a secrecy rate tuple $(R_1,\cdots,R_K)$ is said to be {\em achievable}, if there exists a sequence of $(2^{nR_1},\cdots,2^{nR_K},n)$ codes such that both the average error probability
\begin{eqnarray}
P_e^n=\text{Pr}\left( \cup_{k=1}^K \{(W_1,\cdots,W_k)\neq g_k^n(Y_k^n)\}\right)\label{eq:decode}
\end{eqnarray}
and the leakage rate at each receiver $k$ for $k=1,\ldots,K$
\begin{equation}
\frac{1}{n}I(W_{k+1},\cdots,W_{K};Y_k^n|W_1,\cdots,W_k)\label{eq:security}
\end{equation}
approach zero as $n$ goes to infinity.

Here, the asymptotically small error probability as in \eqref{eq:decode} implies that each receiver $k$ is able to decode messages $W_1,\ldots,W_k$, and asymptotically small leakage rate as in \eqref{eq:security} for each receiver $k$ implies that receiver $k$ is kept ignorant of messages $W_{k+1},\ldots,W_K$. Our goal is to characterize the secrecy capacity region that consists of all achievable rate tuples.



We also consider the $K$-receiver degraded Gaussian broadcast channel, in which
\begin{equation}\label{Gmodel}
Y_k=X+Z_k, \text{ }k=1,\cdots,K,
\end{equation}
where $Z_k$ is a zero mean Gaussian noise variable with variance $N_k$ at receiver $k$. We assume that $0<N_K<N_{K-1}<\ldots<N_1$. The transmitter has an average power constraint $P$. Since the secrecy capacity region only depends on the marginal distribution of the channel input and channel output at each receiver, not on the joint distribution of the channel outputs, changing the correlation of those noise variables has no effects on the secrecy capacity region. Hence, we can adjust the correlation of those noise variables such that the channel outputs at each receiver and the channel input satisfy the same Markov chain as shown in \eqref{eq:markov}.

We further consider the $K-$receiver degraded Gaussian MIMO broadcast channel.
The received signal at receiver $k$ for one channel use is given by
\begin{equation}
\mathbf{Y}_k=\mathbf{X}+\mathbf{Z}_k, \text{ }k=1,\ldots,K,
\end{equation}
where the channel input $\mathbf X$, the channel output $\mathbf Y_k$ and the noise $\mathbf Z_k$ are $r$-dimensional vectors. Furthermore, the noise variables $\mathbf{Z}_k$ are zero-mean Gaussian random vectors with covariance matrices $\mathbf{\Sigma}_k$ for $k=1,\ldots,K$ that satisfy the following order:
\begin{equation}
\mathbf{0}\prec \mathbf{\Sigma}_K \preceq \mathbf{\Sigma}_{K-1}\preceq\cdots\preceq\mathbf{\Sigma}_1.
\end{equation}
The channel input $\mathbf{X}$ is subject to a covariance constraint
\begin{equation}
E[\mathbf{XX}^\top]\preceq \mathbf{S}
\end{equation}
where $\mathbf{S}\succ \mathbf{0}$.
Since the secrecy capacity region does not depend on the correlation across the channel outputs, we can adjust the correlation between the noise vectors such that the channel inputs and channel outputs satisfy the following Markov chain:
\begin{equation}\label{eq:v_markov}
    \mathbf{X}\rightarrow \mathbf{Y}_K\rightarrow \mathbf{Y}_{K-1}\rightarrow\ldots\rightarrow \mathbf{Y}_2\rightarrow \mathbf{Y}_1.
\end{equation}
Hence, the quality of channels gradually degrades from receiver $K$ to receiver $1$.

%

\subsection{Characterization of Secrecy Capacity Region}\label{sec:capacity_result}
In this subsection, we characterize the secrecy capacity region for the model presented in Section \ref{subsec:model}.
The following theorem characterized the secrecy capacity region of the discrete memoryless channel.
\begin{theorem}\label{th:capadmc}
The secrecy capacity region of the degraded broadcast channel with layered decoding and secrecy constraints as described in Section \ref{subsec:model} contains rate tuples $(R_1,\cdots,R_K)$ satisfying
\begin{flalign}\label{bd:dmc}
R_1&\leq I(U_1;Y_1),\nn\\
R_k&\leq I(U_k;Y_k|U_{k-1})-I(U_k;Y_{k-1}|U_{k-1}),\quad\text{for } k=2,\ldots,K-1,\nn\\
R_K&\leq I(X;Y_K|U_{K-1})-I(X;Y_{K-1}|U_{K-1}),
\end{flalign}
for some $P_{U_1U_2\ldots U_{K-1}X}$ such that the following Markov chain holds
\begin{equation}
U_1\rightarrow U_2\rightarrow\ldots\rightarrow U_{K-1}\rightarrow X\rightarrow Y_K\rightarrow\ldots\rightarrow Y_1. \label{Markov}
\end{equation}
\end{theorem}
\begin{remark}
By setting $R_1=0$ and $K=3$, Theorem \ref{th:capadmc} reduces to the results for scenario 2 in \cite{ekrem2012} with each group having a single user for the model in \cite{bagherikaram2008}, and for the example in \cite{Liu2010}.
\end{remark}
\begin{proof}
The proof of the achievability and the proof of converse are provided in Appendices \ref{achiv_dmc} and \ref{conv_dmc}, respectively.
\end{proof}

We here briefly introduce the idea of the achievable scheme, which is based on the stochastic encoding (i.e., random binning) and superposition coding. For each message, we design one layer of codebook. This codebook contains codewords that are divided into a number of bins, where the bin number contains the information of the corresponding message. The receivers that are required to decode the message can tell which bin the codeword is in with a small probability of error, while other receivers (i.e., those with worse channel quality) are kept ignorant of this message. These layers of codebooks are superposed together via superposition coding. The major challenge of the achievability proof arises in the analysis of leakage rates, which is much more involved than the cases with two secure messages studied in \cite{Liu2010,ekrem2012}. In our proof, we develop novel generalization of the analysis provided in \cite{Gamal2012} for the case with one secure message to multiple secure messages. The details can be referred to Appendix \ref{achiv_dmc}.




We next characterize the secrecy capacity region of the degraded Gaussian broadcast channel with layered decoding and secrecy constraints. We note that although the Gaussian channel is a special case of the MIMO channel, we present the result for the Gaussian channel here, because the result for the Gaussian channel is simpler, and hence it is easier to follow the converse proof for this case. This helps the understanding of the more complicated proof of the converse for the MIMO channel.
\begin{theorem}\label{thm:gaussian_channel}
The secrecy capacity region of a $K$-user Gaussian broadcast channel with layered decoding and secrecy constraints as described in Section \ref{subsec:model} contains rate tuples $(R_1,R_2,\ldots,R_K)$ satisfying
\begin{flalign}
R_1& \leq \frac{1}{2}\log\left(\frac{N_1+\sum_{j=1}^{K}P_j}{N_1+\sum_{j=2}^KP_j}\right)\nn\\
R_k& \leq \frac{1}{2}\log\left(\frac{N_k+\sum_{j=k}^{K}P_j}{N_k+\sum_{j=k+1}^KP_j}\right)-\frac{1}{2}\log\left(\frac{N_{k-1}+\sum_{j=k}^{K}P_j}{N_{k-1}+\sum_{j=k+1}^KP_j}\right),\quad\text{for }2\leq k\leq K,
\end{flalign}
for some nonnegative variables $P_1,P_2,\ldots,P_K$ such that $\sum_{k=1}^{K}P_k\leq P$.
\end{theorem}
\begin{proof}
The achievability is based on Theorem \ref{th:capadmc} by setting $(U_1,\ldots,U_K,X)$ to be jointly Gaussian distributed random variables with $U_k\thicksim \mathcal N(0,\sum_{j=1}^kP_j)$.

The converse proof is given in Appendix \ref{conv_gau}.
\end{proof}


We now characterize the secrecy capacity region for the degraded Gaussian MIMO channel with layered decoding and secrecy constraints in the following theorem.
\begin{theorem}\label{MIMO}
The secrecy capacity region of the degraded Gaussian MIMO Broadcast channel with layered decoding and secrecy constraints as described in Section~\ref{subsec:model} contains rate tuples $(R_1,\ldots,R_K)$ satisfying the following inequalities:
\begin{flalign}\label{eq:capacity_mimo}
R_1&\leq \frac{1}{2}\log\frac{|\mathbf{\Sigma}_1+\mathbf{S}|}{|\mathbf{\Sigma}_1+\mathbf{S}_1|}\nn\\
R_k&\leq \frac{1}{2}\log \frac{|\mathbf{\Sigma}_k+\mathbf{S}_{k-1}|}{|\mathbf{\Sigma}_k+\mathbf{S}_k|}-\frac{1}{2}\log\frac{|\mathbf{\Sigma}_{k-1}+\mathbf{S}_{k-1}|}{|\mathbf{\Sigma}_{k-1}+\mathbf{S}_k|},\quad\text{ for }2\leq k\leq K-1\nn\\
R_K&\leq \frac{1}{2}\log\frac{|\mathbf{\Sigma}_K+\mathbf{S}_{K-1}|}{|\mathbf{\Sigma}_K|}-\frac{1}{2}\log\frac{|\mathbf{\Sigma}_{K-1}+\mathbf{S}_{K-1}|}{|\mathbf{\Sigma}_{K-1}|},
\end{flalign}
for some $\mathbf 0\preceq\mathbf{S}_{K-1}\preceq \mathbf{S}_{K-2}\preceq \ldots\preceq \mathbf{S}_2\preceq \mathbf{S}_1\preceq \mathbf{S}$.
\end{theorem}
We note that if the system has only single transmit antenna and receive antenna, then the secrecy capacity region in Theorem \ref{MIMO} reduces to that in Theorem \ref{thm:gaussian_channel}.

We further note that due to the layered secrecy constraints, the major challenge in the converse proof for the secrecy capacity region lies in development of upper bounds in certain recursive structures for three or more consecutive layers of receivers. Our contribution here lies in the construction of a series of covariance matrices representing input resources for layered messages such that the secrecy rates can be upper bounded as the desired recursive forms in terms of these covariance matrices. The details can be referred to Appendix \ref{conv_gau_mimo}.
\begin{proof}
The achievability of the region \eqref{eq:capacity_mimo} follows by choosing the auxiliary random variables $\mathbf U_1,\ldots,\mathbf U_{K-1},\mathbf X$ in \eqref{bd:dmc} to be jointly Gaussian distributed and satisfy the following Markov chain condition:
\begin{equation}
\mathbf{U}_1\rightarrow \mathbf{U}_2\rightarrow \ldots \rightarrow \mathbf{U}_{K-1}\rightarrow \mathbf{X},
\end{equation}
where the covariance of $\mathbf U_k$ is set to be $\mathbf S-\mathbf S_k$ for $k=1,\ldots,K-1$, and the covariance of $\mathbf X$ is set to be $\mathbf S$.

The proof of converse is given in Appendix \ref{conv_gau_mimo}.
\end{proof}

\section{Application to Sharing Multiple Secrets}\label{sec:app}

\begin{figure}[thb]
\begin{center}
\includegraphics[width=4in]{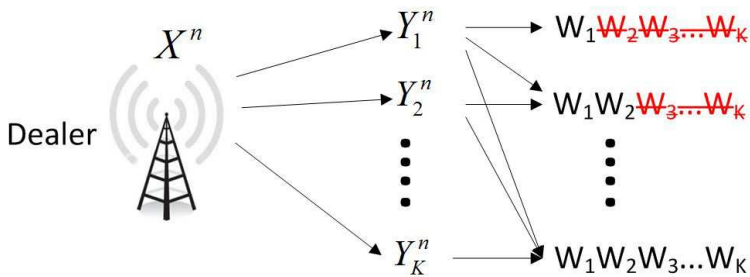}
\caption{The model of secret sharing via a broadcast channel}\label{secretmodel}
\end{center}
\end{figure}

In this section, we apply our result in Section \ref{sec:capacity_result} for the degraded MIMO channel to studying the following problem of sharing multiple secrets. Here, a dealer wishes to share $K$ secrets $W_1,W_2,\ldots,W_K$
with $K$ participants via a broadcast channel (see Fig.~\ref{secretmodel}). The channel input sent by the dealer is denoted by $\mathbf X$ and the channel output received at participant $k$ is denoted by $Y_k$ for $k=1,\ldots,K$.
It is required that participant 1 decodes $W_1$, and participants $1$ and $2$ decode $W_1$ and $W_2$ by sharing their outputs $(Y_1,Y_2)$, but $W_2$ should be kept secure from participant $1$.
Such requirements extend to $k$ participants for $k=1,\ldots,K$ in the sense that participants $1$ to $k$ can recover the first $k$ messages $W_1,\ldots,W_k$ by sharing their outputs $(Y_1,\ldots,Y_k)$,
but the new message $W_k$ should be secure from the first $k-1$ participants. Hence, as one more participant joins the group, one more secret can be recovered, and this new secret is secure from
(and hence cannot be recovered by) a smaller group. The goal is to characterize the secret sharing capacity region, which contains all possible achievable rate tuples $(R_1,R_2,\ldots,R_K)$
for $K$ secrets.

This secret sharing problem involves sharing multiple secrets in a layered fashion, and is challenging to solve using the classical
approach based on number theory. Here, we solve this problem by constructing an equivalent Gaussian MIMO broadcast model as described in Section \ref{subsec:model}.

We assume that the dealer communicates to the participants via a Gaussian MISO broadcast channel corrupted by additive Gaussian noise variables.
The dealer has $K$ antennas and each receiver has one antenna.
The relationship of the channel input from the dealer and the channel outputs at all participants is given by
\begin{equation}
\left(\begin{array}{c}
        Y_1 \\
        \vdots \\
        Y_K
      \end{array}
      \right)=\mathbf H\left(\begin{array}{c}
                              X_1 \\
                              \vdots \\
                              X_K
                            \end{array}\right)+\left(\begin{array}{c}
                                                       Z_1 \\
                                                       \vdots \\
                                                       Z_K
                                                     \end{array}\right)
\end{equation}
where $\mathbf H$ is the $K\times K$ channel matrix, which is assumed to be invertible,
$(Y_1,\ldots,Y_K)$ are channel outputs at the $K$ participants, $(X_1,\ldots X_K)$ are the channel inputs from the $K$ antennas of the dealer, and
$(Z_1,\ldots, Z_K)$ is a random Gaussian vector with the covariance matrix $\mathbf\Sigma$ with each entry $\mathbf\Sigma_{ij}=E[Z_iZ_j]=\sigma^2_{ij}$.
We assume that the dealer's input is subject to a resource constraint, $E[\mathbf X\mathbf X^T]\preceq \mathbf S$.

We note that it is reasonable to assume that $\mathbf H$ is invertible in order to guarantee that each participant's output contains new information compared to other participants so that
new secret can be recovered when this participant joins a group.

We reformulate the above secret sharing model into a degraded MIMO broadcast communication system by designing a virtual receiver for each sharing group of participants.
More specifically, we design a virtual receiver $\mathbf V_k$ for the group of the first $k$ participants, i.e., $\mathbf V_k=(Y_1,\ldots,Y_k)$, for $1\leq k\leq K$.
For technical convenience, we add $K-k$ specially designed outputs $\widetilde Y_{k+1},\ldots,\widetilde Y_K$ to $\mathbf V_k$ so that it contains $K$ components, i.e., the virtual receiver $\mathbf V_k$
has $K$ antennas.
The channel outputs at those $K$ antennas are given by,
\begin{equation}\label{secret_eq}
    \mathbf V_k=\left(
                  \begin{array}{c}
                    Y_1\\
                    \vdots \\
                    Y_k \\
                    \widetilde Y_{k+1} \\
                    \vdots \\
                    \widetilde Y_K \\
                  \end{array}
                \right)=\mathbf H\left(
                                  \begin{array}{c}
                                    X_1 \\
                                    \vdots \\
                                    X_K \\
                                  \end{array}
                                \right)+\left(
                                          \begin{array}{c}
                                            Z_1 \\
                                            \vdots \\
                                            Z_k \\
                                             Z_{k+1}+t\widetilde Z_{k+1}\\
                                            \vdots \\
                                            Z_K+t\widetilde Z_K \\
                                          \end{array}
                                        \right)
\end{equation}
where $\widetilde Z_{k}$, $2\leq k\leq K$, is random Gaussian noise variables with mean zero and variance $\tilde\sigma_{kk}^2 >0$,
and $\widetilde Z_{k}$ is independent from all other random variables. Here,
$t$ is a large enough constant (i.e., $t\rightarrow\infty$), so that $\widetilde Y_{k+1},\ldots,\widetilde Y_K$ are fully corrupted by the noise.
We define a new random Gaussian vector $\mathbf Z_V(k)=(Z_1,\ldots,Z_k,Z_{k+1}+t\widetilde Z_{k+1},\ldots,Z_K+t\widetilde Z_K)^T$ and rewrite \eqref{secret_eq} as
\begin{equation}\label{eqmodel}
    \mathbf V_k=\mathbf H \mathbf X+\mathbf Z_V(k), \text{ for }k=1,\ldots,K.
\end{equation}
Since the channel matrix $\mathbf H$ is invertible, we have
\begin{equation}
    \mathbf H^{-1}\mathbf V_k=\mathbf X+\mathbf H^{-1}\mathbf Z_V(k).
\end{equation}
By treating $\mathbf H^{-1}\mathbf V_k$ as the new channel output $\mathbf V'_k$ at virtual receiver $\mathbf V_k$,
 and define a new random Gaussian noise vector $\mathbf Z_V'(k)=\mathbf H^{-1}\mathbf Z_V(k)$, we have
\begin{equation}
    \mathbf V'_k=\mathbf X+\mathbf Z_V'(k),
\end{equation}
which is equivalent to the model in \eqref{eqmodel}.

We now state a lemma that provides the order of the covariance matrices of $\mathbf Z'_V(k)$, denoted by $\mathbf \Sigma_V'(k)$, for $1\leq k\leq K$.

\begin{lemma}\label{lemmaorder}
Let $\mathbf Z_V'(k)$, $1\leq k\leq K$, be random Gaussian vectors defined as above. The covariance matrices of $\mathbf Z_V'(k)$  satisfy the following ordering property:
\begin{equation}\label{order_matrix}
    \mathbf \Sigma_V'(1)\succeq \mathbf \Sigma_V'(2)\succeq\ldots\succeq \mathbf \Sigma_V'(K).
\end{equation}
\end{lemma}
\begin{proof}
For any $1\leq k\leq K-1$, the covariance matrix of $\mathbf Z_V'(k)$ is given by
\begin{equation}
\begin{split}
    \mathbf \Sigma_V'(k)&=Cov(\mathbf H^{-1}\mathbf Z_V(k))\\
    &=E[\mathbf H^{-1}\mathbf Z_V(k)\mathbf Z_V(k)^T(\mathbf H^{-1})^T]\\
    &=\mathbf H^{-1}E[\mathbf Z_V(k)\mathbf Z_V(k)^T](\mathbf H^{-1})^T\\
    &=\mathbf H^{-1}\left(
        \begin{array}{cccccc}
          \sigma_{11}^2 & \cdots & \cdots & \cdots &\cdots & \sigma_{1K}^2 \\
          \vdots & \ddots &   &   &   & \vdots \\
          \vdots &   & \sigma_{kk}^2 &   &   & \vdots \\
          \vdots &   &   & \sigma_{k+1,k+1}^2+t^2\tilde{\sigma}_{k+1,k+1}^2 &   & \vdots \\
          \vdots &   &   &   & \ddots & \vdots \\
          \sigma_{K1}^2 & \ldots & \ldots & \ldots & \ldots & \sigma_{KK}^2+t^2\tilde{\sigma}_{KK}^2 \\
        \end{array}
      \right)(\mathbf H^{-1})^T
\end{split}
\end{equation}

Hence,
\begin{equation}\label{mtr}
\begin{split}
    &\mathbf \Sigma_V'(k)-\mathbf \Sigma_V'(k+1)\\
    =&\mathbf H^{-1}\left(
        \begin{array}{ccc}
          \mathbf 0 & \mathbf 0 & \mathbf 0 \\
          \mathbf 0 & t^2\tilde{\sigma}_{k+1,k+1}^2 & \mathbf 0 \\
          \mathbf 0 & \mathbf 0 & \mathbf 0 \\
        \end{array}
      \right)(\mathbf H^{-1})^T
\end{split}
\end{equation}
It is clear that $\mathbf \Sigma_V'(k)-\mathbf \Sigma_V'(k+1)$ is a positive semi-definite matrix and hence $\mathbf \Sigma_V'(k)\succeq \mathbf \Sigma_V'(k+1)$, for $1\leq k\leq K-1$. This concludes the proof.
\end{proof}

Therefore, by designing virtual receivers, we reformulate the problem of secret sharing via the MISO broadcast channel into the problem of secure communication over the degraded MIMO broadcast channel described in Section \ref{subsec:model}. It can also be seen that the requirements of the secret sharing problem is equivalent to the layered decoding and secrecy requirements for the communication problem. It is also due to the secret sharing requirements and Lemma \ref{lemmaorder}, the degradedness condition in the equivalent MIMO channel naturally holds. Thus, the secret sharing capacity region equals the secrecy capacity region of the degraded MIMO broadcast channel.
Thus, applying Theorem \ref{MIMO} we obtain the following secret sharing capacity region.
\begin{corollary}
The capacity region for the secret sharing problem described above contains rate tuples $(R_1,R_2,\ldots,R_K)$ satisfying
\begin{flalign}\label{secret_capacityregion}
        R_1&\leq \frac{1}{2}\log\frac{|\mathbf{\Sigma}_V'(1)+\mathbf{S}|}{|\mathbf{\Sigma}_V'(1)+\mathbf{S}_1|} \nn\\
        R_k&\leq \underset{t\rightarrow\infty}{\lim}\frac{1}{2}\log \frac{|\mathbf{\Sigma}_V'(k)+\mathbf{S}_{k-1}|}{|\mathbf{\Sigma}_V'(k)+\mathbf{S}_k|}-\frac{1}{2}\log\frac{|\mathbf{\Sigma}_V'(k-1)+\mathbf{S}_{k-1}|}{|\mathbf{\Sigma}_V'(k-1)+\mathbf{S}_k|}, \quad \text{ for }2\leq k\leq K-1, \nn\\
        R_K&\leq \underset{t\rightarrow\infty}{\lim} \frac{1}{2}\log\frac{|\mathbf{\Sigma}_V'(K)+\mathbf{S}_{K-1}|}{|\mathbf{\Sigma}_V'(K)|}-\frac{1}{2}\log\frac{|\mathbf{\Sigma}_V'(K-1)+\mathbf{S}_{K-1}|}{|\mathbf{\Sigma}_V'(K-1)|},
\end{flalign}
for some $\mathbf 0\preceq\mathbf{S}_{K-1}\preceq \mathbf{S}_{K-2}\preceq \ldots\preceq \mathbf{S}_2\preceq \mathbf{S}_1\preceq \mathbf{S}$.
\end{corollary}

\section{Conclusion}\label{sec:con}

In this paper, we have proposed a novel approach based on information theoretic secrecy to solving secret sharing problems. The basic idea is to reformulate the secret sharing problem into a secure communication problem, and then apply techniques for the latter case to solving the secret sharing problem. In order to illustrate the basic idea, we have first studied the classic problem of sharing one secret among a set of participants, and provided a solution by reformulating the secret sharing system into an equivalent compound wiretap channel with multiple legitimate receivers and multiple eavesdroppers. We have then demonstrated the power of our approach by solving a more complicated problem of sharing multiple secrets with layered sharing requirements. We have characterized the secret sharing capacity region by reformulating the problem into the problem of the degraded broadcast MIMO channel with layered decoding and secrecy, for which we have characterized the secrecy capacity region. Our approach can be generally applicable to solving various secret sharing problems, which can be difficult using traditional number theoretic tools. For example, various multiple-secret sharing problems can be reformulated into secure communication problems with multiple confidential messages via compound MIMO broadcast channels, and hence the information theoretic techniques and results developed in existing literature, e.g., \cite{Weingarten2009,Mari2009,Mohammad2009}, can be applied to solving these secret sharing problems.

\vspace{0.5in}
\appendix

\noindent {\Large \textbf{Appendix}}
\section{Achievability Proof of Theorem \ref{th:capadmc}}\label{achiv_dmc}

The achievability proof is based on stochastic encoding and superposition coding. We use random codes and fix a joint probability distribution $P_{U_1\cdots U_{K-1}X}$ satisfying the Markov chain condition given in~\eqref{Markov}.
Let $T_{\epsilon}^{n}(P_{U_1\ldots U_{K-1}XY_1\ldots Y_K})$ denote the strongly
jointly $\epsilon$-typical set based on the fixed distribution.


{\em Random codebook generation:}
In the following achievability proof, for notational convenience, we write $X$ as $U_K$, i.e., $P_{U_1\cdots U_{K-1}X}=P_{U_1\cdots U_{K}}$.
\begin{list}{$\bullet$}{\topsep=0ex \leftmargin=6.5mm \rightmargin=0mm \itemsep=0mm}

\item Generate $2^{nR_1}$ independent and identically distributed
(i.i.d.) $u^n_1$ with distribution $\prod_{i=1}^{n}p(u_{1,i})$.
Index these codewords as $u^n_1(w_1)$, $w_1 \in [1,2^{nR_1}]$.
%

\item For each $u_{k-1}^n(w_1,w_2,l_2,\cdots,w_{k-1},l_{k-1})$, $k=2,\cdots, K$, generate $2^{n\widetilde{R}_k}$ i.i.d.\ sequences $u_k^n$ with distribution $\prod_{i=1}^{n}p(u_{k,i}|u_{k-1,i})$. Partition these sequences into $2^{nR_k}$ bins, each with $2^{n(\widetilde{R}_k-R_k)}$ sequences. We use $w_k\in[1:2^{nR_k}]$ to denote the bin index, and $l_k\in[1:2^{n(\widetilde{R}_k-R_k)}]$ to denote the index within each bin. Hence each $u_k^n$ is indexed by $(w_1,w_2,l_2,\cdots,w_k,l_k)$.


\end{list}
The chosen codebook is revealed to the transmitter and all receivers.


{\em Encoding:} To send a message tuple $(w_1,w_2,\ldots,w_K)$, for each $2\leq k\leq K$, the encoder randomly generate $l_k$ $\in [1:2^{n(\widetilde{R}_k-R_k)}]$ based on a uniform distribution. The transmitter then sends $u_K^n(w_1,w_2,l_2,\cdots,w_K,l_K)$.

{\em Decoding:} For $k=1,\ldots,K$, receiver $k$ claims that $(\hat{w}_1,\cdots,\hat{w}_k)$ is sent, if there exists a unique tuple $(\hat{w}_1,\hat{w}_2,\hat{l}_2,\cdots,\hat{w}_k,\hat{l}_k)$ such that
\begin{flalign}
&(u_1^n(\hw_1), u_2^n(\hw_1,\hw_2,\hl_2),\ldots, u_k^n(\hat{w}_1,\hat{w}_2,\hat{l}_2,\cdots,\hat{w}_k,\hat{l}_k),y_k^n) \in T_{\epsilon}^{n}(P_{U_1\ldots U_k Y_k}).
\end{flalign}
Otherwise, it declares an error.

{\em Analysis of error probability:} By the law of large numbers and the packing lemma, it can be shown that if the following inequalities are satisfied, receiver $k$ (for $k=1,\ldots,K$) can decode messages $w_1,w_2,\ldots,w_k$ with a vanishing error probability:
\begin{equation}\label{pe}
\begin{split}
R_1&\leq I(U_1;Y_1),\\
\widetilde{R}_k&\leq I(U_k;Y_k|U_{k-1}),\quad\text{for } 2\leq k\leq K.
\end{split}
\end{equation}

{\em Analysis of leakage rate:}
We first compute an average of the leakage rate over the random codebook ensemble as follows. For convenience, we let $W^k=(W_1,\ldots,W_k)$, $W_{k+1}^K=(W_{k+1},\ldots,W_K)$.
\begin{flalign}
&I(W_{k+1}^K;Y_k^n|W^k,\cC)\nn\\
&=I(W^K,L^K;Y_k^n|\cC)-I(W^k,L^K;Y_k^n|W_{k+1}^K,\cC)+H(W^k|Y_k^n,\cC)-H(W^k|Y_k^n,W_{k+1}^K,\cC)\nn\\
&\overset{(a)}{\leq} I(W^K,L^K;Y_k^n|\cC)-I(W^k,L^K;Y_k^n|W_{k+1}^K,\cC)+n\epsilon_n\nn\\
&\overset{(b)}{\leq} I(U_K^n;Y_k^n|\cC)-I(W^k,L^K;Y_k^n|W_{k+1}^K,\cC)+n\epsilon_n\nn\\
&=I(U_K^n;Y_k^n|\cC)-H(W^k,L^K|W_{k+1}^K,\cC)+H(W^k,L^K|Y_k^n,W_{k+1}^K,\cC)+n\epsilon_n,
\end{flalign}
where step (a) follows from Fano's inequality, step (b) follows from the Markov chain $(W^K, L^K)\rightarrow (U_K^n,\cC)\rightarrow Y_k^n$.

We bound the above three terms one by one. For the first term, we have
\begin{flalign}
&I(U_K^n;Y_k^n|\cC)\nn\\
&\overset{(a)}{=}I(U_k^n,U_K^n;Y_k^n|\cC)\nn\\
&=I(U_k^n;Y_k^n|\cC)+I(U_K^n;Y_k^n|U_k^n,\cC)\nn\\
&\leq H(U_k^n|\cC)+I(U_K^n;Y_k^n|U_k^n,\cC)\nn\\
&=n\sum_{j=1}^k\widetilde{R}_j+H(Y_k^n|U_k^n,\cC)-H(Y_k^n|U_K^n,U_k^n,\cC)\nn\\
&=n\sum_{j=1}^k\widetilde{R}_j+\sum_{j=1}^n H(Y_{k,j}|U_k^n,Y_k^{j-1},\cC)-\sum_{j=1}^n H(Y_{k,j}|U_K^n,U_k^n,Y_k^{j-1},\cC)\nn\\
&\overset{(b)}{\leq}n\sum_{j=1}^k\widetilde{R}_j+\sum_{j=1}^nH(Y_{k,j}|U_{k,j})-\sum_{j=1}^nH(Y_{k,j}|U_{K,j})\nn\\
&= n\sum_{j=1}^k\widetilde{R}_j+nH(Y_k|U_k)-nH(Y_k|U_K)\nn\\
&=n\sum_{j=1}^k\widetilde{R}_j+nI(U_K;Y_k|U_k).
\end{flalign}
Where (a) follows from the Markov chain $U_k^n\rightarrow U_K^n\rightarrow Y_k^n$, (b) follows from the fact that $H(Y_{k,j}|U_k^n,Y_k^{j-1},\cC)\leq H(Y_{k,j}|U_{k,j})$ and from the Markov chain $(U_k^n, U_K^{j-1},U_{K,j+1}^n,Y_k^{j-1},\cC)\rightarrow U_{K,j}\rightarrow Y_{k,j}$.

For the second term, due to the independence of $W_1,\cdots,W_K$ and $L_1,\cdots,L_K$, we have
\begin{flalign}
&H(W^k,L^K|W_{k+1}^K,\cC)=\sum_{j=1}^kn\widetilde{R}_j+\sum_{j=k+1}^Kn(\widetilde{R}_j-R_j).
\end{flalign}
We now bound the last term as follows.
\begin{flalign}
&H(W^k,L^K|Y_k^n,W_{k+1}^K,\cC)\nn\\
&=H(W^k|Y_k^n,W_{k+1}^K,\cC)+H(L^K|Y_k^n,W^K,\cC)\nn\\
&\overset{(a)}{\leq} H(L_{k+1}^K|Y_k^n,W^K,L^k,\cC)+2n\epsilon_n\nn\\
&=\sum_{j=k+1}^KH(L_j|Y_k^n,W^K,L^{j-1},\cC)+2n\epsilon_n\nn\\
&\overset{(b)}{=}\sum_{j=k+1}^KH(L_j|Y_k^n,W^K,L^{j-1},U_{j-1}^n,\cC)+2n\epsilon_n\nn\\
&\leq\sum_{j=k+1}^KH(L_j|Y_k^n,U_{j-1}^n,W_j)+2n\epsilon_n\nn\\
&\overset{(c)}{\leq}\sum_{j=k+1}^Kn(\tilde{R}_j-R_j-I(U_j;Y_k|U_{j-1}))+n\epsilon_n'\nn\\
&\overset{(d)}{=}\sum_{j=k+1}^Kn(\tilde{R}_j-R_j)-I(U_K;Y_k|U_k)+n\epsilon_n',
\end{flalign}
where (a) follows from the chain rule and Fano's inequality, (b) follows from the fact that $U_{j-1}^n$ is a function of $(\cC, W^{j-1},L^{j-1})$, and (c) follows due to Lemma \ref{lemma:nit} with the condition that $\tilde{R}_j-R_j\geq I(U_j;Y_k|U_{j-1})$, and (d) follows from the Markov chain $U_1\rightarrow U_2\rightarrow \cdots\rightarrow U_K\rightarrow Y_k$.


\begin{lemma}\label{lemma:nit}
If $\tilde{R}_j-R_j\geq I(U_j;Y_k|U_{j-1})$ for $k+1\leq j\leq K$, then
\[\frac{1}{n}H(L_j|Y_k^n,U_{j-1}^n,W_j)\leq \widetilde{R}_j-R_j-I(U_j;Y_k|U_{j-1})+\epsilon''_n.\]
\end{lemma}
\begin{proof}
See Appendix \ref{proofnit}.
\end{proof}

Combining the analysis of the three terms together, we conclude that as $n\rightarrow\infty$ for $1\leq k\leq K-1$, $\frac{1}{n}I(W_{k+1}^K;Y_k^n|W^k,\cC)\rightarrow 0$, if
\begin{flalign}\label{eqv}
\widetilde{R}_k-R_k\geq I(U_k;Y_{k-1}|U_{k-1}), \quad\quad \text{ for }2\leq k\leq K.
\end{flalign}

It is also clear that the sum of the error probability and the leakage rates averaged over the codebook ensemble converges to zero as $n\rightarrow \infty$. Hence, there exists one codebook such that the error probability and the leakage rate converge to zero as $n\rightarrow \infty$.

Combining the bounds in \eqref{pe} and \eqref{eqv}, we obtain that the rate tuple $(R_1,\cdots,R_K)$ is achievable if
\begin{flalign}
R_1&\leq I(U_1;Y_1),\nn\\
R_k&\leq I(U_k;Y_k|U_{k-1})-I(U_k;Y_{k-1}|U_{k-1}),\text{ for }2\leq k\leq K.
\end{flalign}


\section{Proof of Lemma \ref{lemma:nit}}\label{proofnit}
We first bound $\frac{1}{n}H(L_j|Y_k^n,U_{j-1}^n,w_j)$ for any $w_j$, and hence, $\frac{1}{n}H(L_j|Y_k^n,U_{j-1}^n,W_j)$ is bounded.

Fix $L_j=l_j$ and a joint typical sequence $(u_{j-1}^n,y_k^n)\in T_\epsilon^{(n)}(U_{j-1},Y_k)$. We define
\begin{flalign}
N(w_j,l_j,u_{j-1}^n,y_k^n):=|\{\tilde{l}_j\neq l_j:(U_j^n(w_j,\tilde{l}_j),u_{j-1}^n,y_k^n)\in T_\epsilon^{(n)}\}|.
\end{flalign}
It can be shown that the expectation and variance of $N$ satisfy the following inequalities:
\begin{flalign}\nn
2^{n(\tilde{R}_j-R_j)-nI(U_j;Y_k|U_{j-1})-n\delta_n(\epsilon)-n\epsilon_n}\leq E(N(w_j,l_j,u_{j-1}^n,y_k^n))\leq 2^{n(\tilde{R}_j-R_j)-nI(U_j;Y_k|U_{j-1})+n\delta_n(\epsilon)-n\epsilon_n}\\
Var(N(w_j,l_j,u_{j-1}^n,y_k^n))\leq 2^{n(\tilde{R}_j-R_j)-nI(U_j;Y_k|U_{j-1})+n\delta_n(\epsilon)-n\epsilon_n},
\end{flalign}
where $\delta_n(\epsilon),\epsilon_n\rightarrow 0$ as $n\rightarrow \infty$.

We next define the random event,
\[\varepsilon(w_j,l_j,u_{j-1}^n,y_k^n):=\{N(w_j,l_j,u_{j-1}^n,y_k^n)\geq 2^{n(\tilde{R}_j-R_j-I(U_j;Y_k|U_{j-1})+\delta_n(\epsilon)-\epsilon_n/2)+1}\}.\]
Using Chebyshev's inequality, we obtain
\begin{flalign}\label{var}
&P\left(\varepsilon(w_j,l_j,u_{j-1}^n,y_k^n)\right)\nn\\
&=P\left(N(w_j,l_j,u_{j-1}^n,y_k^n)\geq 2^{n(\tilde{R}_j-R_j-I(U_j;Y_k|U_{j-1})+\delta_n(\epsilon)-\epsilon_n/2)+1}\right)\nn\\
&\leq P\left(N(w_j,l_j,u_{j-1}^n,y_k^n)\geq E(N(w_j,l_j,u_{j-1}^n,y_k^n))+2^{n(\tilde{R}_j-R_j-I(U_j;Y_k|U_{j-1})+\delta_n(\epsilon)-\epsilon_n/2)}\right)\nn\\
&\leq P\left(|N(w_j,l_j,u_{j-1}^n,y_k^n)-E(N(w_j,l_j,u_{j-1}^n,y_k^n))|\geq 2^{n(\tilde{R}_j-R_j-I(U_j;Y_k|U_{j-1})+\delta_n(\epsilon)-\epsilon_n/2)}\right)\nn\\
&\leq \frac{Var(N(w_j,l_j,u_{j-1}^n,y_k^n))}{2^{2n(\tilde{R}_j-R_j-I(U_j;Y_k|U_{j-1})+\delta_n(\epsilon)-\epsilon_n/2)}}\nn\\
&\leq\frac{1}{2^{n(\tilde{R}_j-R_j-I(U_j;Y_k|U_{j-1})+\delta_n(\epsilon))}}
\end{flalign}
which goes to zero as $n\rightarrow \infty$ if
$\tilde{R}_j-R_j\geq I(U_j;Y_k|U_{j-1})$. This implies that
\[P\bigg(\varepsilon(w_j,l_j,u_{j-1}^n,y_k^n)\bigg)\rightarrow 0\]
as $n\rightarrow \infty$.

For each message $w_j$, we define the following random variable and event:
\begin{flalign}
N(w_j):=|\{\tilde{l}_j:(U_j^n(w_j,\tilde{l}_j),Y_k^n,U_{j-1}^n)\in T_\epsilon^{(n)},\tilde{l}_j\neq L_j\}|\nn\\
\varepsilon(w_j):=\{N(w_j)\geq 2^{n(\tilde{R}_j-R_j-I(U_j;Y_k|U_{j-1})+\delta_n(\epsilon)-\epsilon_n/2)+1}\}\nn
\end{flalign}
Finally, define the indicator random variable $E(w_j):=0$ if $(U_j^n(w_j,L_j),Y_k^n,U_{j-1}^n)\in T_\epsilon^{(n)}$ and $\varepsilon(w_j)^c$ occurs; and $E(w_j):=1$, otherwise.
Therefore, we have
\begin{flalign}\label{eq43}
P\left(E(w_j)=1\right)&\leq P\left( (U_j^n(w_j,L_j),U_{j-1}^n,Y_k^n)\notin T_\epsilon^{(n)}\right)+P\left(\varepsilon(w_j)\right).
\end{flalign}
It is clear that the first term in \eqref{eq43} goes to zero as $n\rightarrow \infty$. For the second term in \eqref{eq43}, we have
\begin{flalign}
&P\left(\varepsilon(w_j)\right)\nn\\
&\leq \sum_{(u_{j-1}^n,y_k^n)\in T_\epsilon^{(n)}}P\left(u_{j-1}^n,y_k^n\right)P\left(\varepsilon(w_j)|u_{j-1}^n,y_k^n\right)+P\left((U_{j-1}^n,Y_k^n)\notin T_\epsilon^{(n)}\right)\nn\\
&=\sum_{(u_{j-1}^n,y_k^n)\in T_\epsilon ^{(n)}}\sum_{l_j}P\left(u_{j-1}^n,y_k^n\right)P\left(l_j|u_{j-1}^n,y_k^n\right)P\left(\varepsilon(w_j)|u_{j-1}^n,y_k^n,l_j\right)+P\left((U_{j-1}^n,Y_k^n)\notin T_\epsilon^{(n)}\right)\nn\\
&\rightarrow 0,\quad\text{ if }\tilde{R}_j-R_j\geq I(U_j;Y_k|U_{j-1}).
\end{flalign}
Therefore,
\begin{flalign}\label{lemma1}
&H(L_j|w_j,U_{j-1}^n,Y_k^n)\nn\\
&\leq H(L_j,E(w_j)|w_j,U_{j-1}^n,Y_k^n)\nn\\
&\leq H(E(w_j))+H(L_j|w_j,U_{j-1}^n,Y_k^n,E(w_j))\nn\\
&\leq 1+P\left(E(w_j)=1\right)H(L_j|w_j,Y_k^n,U_{j-1}^n,E(w_j)=1)+H(L_j|w_j,Y_k^n,U_{j-1}^n,E(w_j)=0)\nn\\
&\leq 1+P\left(E(w_j)=1\right)n(\tilde{R}_j-R_j)+\log 2^{n(\tilde{R}_j-R_j-I(U_j;Y_k|U_{j-1})+\delta(\epsilon)-\epsilon/2)+1} \nn \\
& =1+n(\tilde{R}_j-R_j)P\left(E(w_j)=1\right)+n(\tilde{R}_j-R_j-I(U_j;Y_k|U_{j-1})+\delta(\epsilon)-\epsilon/2)+1
\end{flalign}
Following from \eqref{lemma1}, we obtain,
\begin{flalign}
\underset{n\rightarrow \infty}{\lim}\frac{1}{n}H(L_j|w_j,U_{j-1}^n,Y_k^n)\leq \tilde{R}_j-R_j-I(U_j;Y_k|U_{j-1})+\delta'(\epsilon),
\end{flalign}
where $\delta'(\epsilon)\rightarrow 0$ as $n\rightarrow \infty$.
This concludes the proof.

\section{Converse Proof of Theorem \ref{th:capadmc}}\label{conv_dmc}

By Fano's inequality and the secrecy requirements, we have the following inequalities
\begin{flalign}
H(W_k|Y_k^n)\leq n\epsilon_n&,\quad\text{for }1\leq k\leq K\\
\frac{1}{n}I(W_{k+1},\ldots,W_K;Y_k^n|W_1,\ldots,W_{k})\leq\epsilon_n&,\quad \text{for }1\leq k\leq K-1.\label{con1}
\end{flalign}

We let $Y_k^{i-1}:=(Y_{k,1},\ldots,Y_{k,i-1})$, and $Y_{k,i+1}^n:=(Y_{k,i+1},\ldots,Y_{k,n})$.
We set \small{$U_{k,i}:=\{W_1,\ldots,W_k,\\Y_k^{i-1},Y_{k-1,i+1}^n\}$} for $k=1,\ldots,K$ where $Y_0^n=\phi$. It is easy to verify that $(U_{1,i},\ldots,U_{K-1,i},X_i)$ satisfy the following Markov chain condition:
\begin{equation}\label{eq51}
U_{1,i}\rightarrow U_{2,i}\rightarrow\ldots\rightarrow U_{K-1,i}\rightarrow X_i\rightarrow Y_{K,i}\rightarrow\ldots\rightarrow Y_{1,i},\quad\text{for }1\leq i\leq n.
\end{equation}

We first bound the rate $R_1$. Since there is no secrecy constraint for $W_1$, following the standard steps, we obtain the following bound:
\begin{flalign}
nR_1&=H(W_1)=I(W_1;Y_1^n)+H(W_1|Y_1^n)\nn\\
&\leq I(W_1;Y_1^n)+n\epsilon_n=\sum_{i=1}^nI(W_1;Y_{1i}|Y_{1}^{i-1})+n\epsilon_n\nn\\
&\leq \sum_{i=1}^nI(W_1,Y_{1}^{i-1};Y_{1i})+n\epsilon_n=\sum_{i=1}^nI(U_{1,i};Y_{1,i})+n\epsilon_n.
\end{flalign}

For the message $W_k$, $2\leq k\leq K$, we derive the following bound:
\begin{flalign}
nR_k=&H(W_k)=H(W_k|W^{k-1})\nn\\
=&I(W_k;Y_k^n|W^{k-1})+H(W_k|W^{k-1},Y_k^n)\nn\\
\overset{(a)}{\leq}&I(W_k;Y_k^n|W^{k-1})+n\epsilon_n\nn\\
\overset{(b)}{\leq}& I(W_k;Y_k^n|W^{k-1})+2n\epsilon_n-I(W_k;Y_{k-1}^n|W^{k-1})\nn\\
=&\sum_{i=1}^nI(W_k;Y_{k,i}|W^{k-1},Y_k^{i-1})+2n\epsilon_n-\sum_{i=1}^nI(W_k;Y_{k-1,i}|W^{k-1},Y_{k-1,i+1}^n)\nn\\
=&\sum_{i=1}^n\bigg[I(W_k,Y_{k-1,i+1}^n;Y_{k,i}|W^{k-1},Y_k^{i-1})-I(W_k,Y_k^{i-1};Y_{k-1,i}|W^{k-1},Y_{k-1,i+1}^n)\nn\\
&-I(Y_{k-1,i+1}^n;Y_{k,i}|W^k,Y_k^{i-1})+I(Y_k^{i-1};Y_{k-1,i}|W^k,Y_{k-1,i+1}^n)\bigg]+2n\epsilon_n\nn\\
\overset{(c)}{=}&\sum_{i=1}^n\bigg[I(Y_{k-1,i+1}^n;Y_{k,i}|W^{k-1},Y_k^{i-1})+I(W_k;Y_{k,i}|W^{k-1},Y_k^{i-1},Y_{k-1,i+1}^n)\nn\\
&-I(Y_k^{i-1};Y_{k-1,i}|W^{k-1},Y_{k-1,i+1}^n)-I(W_k;Y_{k-1,i}|W^{k-1},Y_k^{i-1},Y_{k-1,i+1}^n)\bigg]+2n\epsilon_n\nn\\
\overset{(d)}{=}&\sum_{i=1}^n\bigg[I(W_k;Y_{k,i}|W^{k-1},Y_k^{i-1},Y_{k-1,i+1}^n)-I(W_k;Y_{k-1,i}|W^{k-1},Y_k^{i-1},Y_{k-1,i+1}^n)\bigg]+2n\epsilon_n\nn\\
=&\sum_{i=1}^n\bigg[I(W_k,Y_k^{i-1},Y_{k-1,i+1}^n;Y_{k,i}|W^{k-1})-I(Y_k^{i-1},Y_{k-1,i+1}^n;Y_{k,i}|W^{k-1})\nn\\
&-I(W_k,Y_k^{i-1},Y_{k-1,i+1}^n;Y_{k-1,i}|W^{k-1})+I(Y_k^{i-1},Y_{k-1,i+1}^n;Y_{k-1,i}|W^{k-1})\bigg]+2n\epsilon_n\nn\\
=&\sum_{i=1}^n\bigg[I(W_k,Y_k^{i-1},Y_{k-1,i+1}^n,Y_{k-1}^{i-1},Y_{k-2,i+1}^n;Y_{k,i}|W^{k-1})\nn\\
&-I(W_k,Y_k^{i-1},Y_{k-1,i+1}^n,Y_{k-1}^{i-1},Y_{k-2,i+1}^n;Y_{k-1,i}|W^{k-1})\nn\\
&-I(Y_k^{i-1},Y_{k-1,i+1}^n;Y_{k,i}|W^{k-1})+I(Y_k^{i-1},Y_{k-1,i+1}^n;Y_{k-1,i}|W^{k-1})\bigg]+2n\epsilon_n\nn
\end{flalign}
\begin{flalign}\label{converse}
=&\sum_{i=1}^n\bigg[I(W_k,Y_k^{i-1},Y_{k-1,i+1}^n;Y_{k,i}|W^{k-1},Y_{k-1}^{i-1},Y_{k-2,i+1}^n)\nn\\
&-I(W_k,Y_k^{i-1},Y_{k-1,i+1}^n;Y_{k-1,i}|W^{k-1},Y_{k-1}^{i-1},Y_{k-2,i+1}^n)+I(Y_{k-1}^{i-1},Y_{k-2,i+1}^n;Y_{k,i}|W^{k-1})\nn\\
&-I(Y_{k-1}^{i-1},Y_{k-2,i+1}^n;Y_{k-1,i}|W^{k-1})-I(Y_k^{i-1},Y_{k-1,i+1}^n;Y_{k,i}|W^{k-1})\nn\\
&+I(Y_k^{i-1},Y_{k-1,i+1}^n;Y_{k-1,i}|W^{k-1})\bigg]+2n\epsilon_n\nn\\
\overset{(e)}{=}&\sum_{i=1}^n\bigg[I(W_k,Y_k^{i-1},Y_{k-1,i+1}^n;Y_{k,i}|W^{k-1},Y_{k-1}^{i-1},Y_{k-2,i+1}^n)\nn\\
&-I(W_k,Y_k^{i-1},Y_{k-1,i+1}^n;Y_{k-1,i}|W^{k-1},Y_{k-1}^{i-1},Y_{k-2,i+1}^n)\nn\\
&-I(Y_k^{i-1},Y_{k-1,i+1}^n;Y_{k,i}|W^{k-1},Y_{k-1}^{i-1},Y_{k-2,i+1}^n)\nn\\
&+I(Y_k^{i-1},Y_{k-1,i+1}^n;Y_{k-1,i}|W^{k-1},Y_{k-1}^{i-1},Y_{k-2,i+1}^n)\bigg]+2n\epsilon_n\nn\\
\overset{(f)}{\leq}&\sum_{i=1}^n\bigg[I(W_k,Y_k^{i-1},Y_{k-1,i+1}^n;Y_{k,i}|W^{k-1},Y_{k-1}^{i-1},Y_{k-2,i+1}^n)\nn\\
&-I(W_k,Y_k^{i-1},Y_{k-1,i+1}^n;Y_{k-1,i}|W^{k-1},Y_{k-1}^{i-1},Y_{k-2,i+1}^n)\bigg]+2n\epsilon_n \nn\\
=&\sum_{i=1}^n\bigg[I(U_{k,i};Y_{k,i}|U_{k-1,i})-I(U_{k,i};Y_{k-1,i}|U_{k-1,i})\bigg]+2n\epsilon_n
\end{flalign}
%
where $(a)$ follows from Fano's inequality, $(b)$ follows from \eqref{con1}, i.e., the secrecy constraint,
$(c)$ and $(d)$ follow from the sum identity property in \cite[Lemma 7]{Csiszar1978}, and (e) and (f) follows from the degradedness condition \eqref{eq:markov}.

For $k=K$, we further derive \eqref{converse},
\begin{flalign}
n R_K&\leq\sum_{i=1}^nI(U_{K,i};Y_{K,i}|U_{K-1,i})-I(U_{K,i};Y_{K-1,i}|U_{K-1,i})+2n\epsilon_n\nn\\
=&\sum_{i=1}^n I(U_{K,i},X_i;Y_{K,i}|U_{K-1,i})-I(U_{K,i},X_i;Y_{K-1,i}|U_{K-1,i})\nn\\
&-I(X_i;Y_{K,i}|U_{K,i})+I(X_i;Y_{K-1,i}|U_{K,i})+2n\epsilon_n \\
\leq&\sum_{i=1}^n I(X_i;Y_{K,i}|U_{K-1,i})-I(X_i;Y_{K-1,i}|U_{K-1,i})+2n\epsilon_n ,\nn
\end{flalign}
where the last step follows from \eqref{eq51} and \eqref{eq:markov}.
The proof of the converse is completed by defining a uniformly distributed random variable $Q\in\{1,\cdots,n\}$, and setting $U_k\triangleq (Q,U_{k,Q})$, $Y_k\triangleq Y_{k,Q}$, for $k\in[1:K]$, and $X\triangleq (Q,X_{Q})$.

\section{Proof of Converse of for Theorem \ref{thm:gaussian_channel}}\label{conv_gau}
We continue the bounds proved for the discrete memoryless channel in their single-letter forms.
%
We first bound $R_1$ as follows:
\begin{equation}\label{converse1}
\begin{split}
R_1\leq& I(U_1;Y_1)+\epsilon_n\leq\frac{1}{2}\log 2\pi e(P+N_1)-h(Y_1|U_1)+\epsilon_n,
\end{split}
\end{equation}
where $h(Y_1|U_1)$ will be bounded later.

Following from \eqref{converse}, we have,
\begin{flalign}\label{eq57}
R_2\leq& I(U_2;Y_2|U_1)-I(U_2;Y_1|U_1)+2\epsilon_n \nn\\
=&h(Y_2|U_1)-h(Y_1|U_1)-(h(Y_2|U_1,U_2)-h(Y_1|U_1,U_2))+2\epsilon_n
\end{flalign}
It can be shown as (140) in \cite{ekrem2012} that
\begin{equation}
\frac{1}{2}\log\frac{N_2}{N_1}\leq h(Y_2|U_1)-h(Y_1|U_1)\leq \frac{1}{2}\log\frac{P+N_2}{P+N_1}.
\end{equation}
Hence, there must exist an $\alpha_1$, $0\leq \alpha_1\leq1$ such that
\begin{equation}\label{ieq6}
h(Y_2|U_1)-h(Y_1|U_1)=\frac{1}{2}\log\frac{(1-\alpha_1)P+N_2}{(1-\alpha_1)P+N_1}.
\end{equation}
Due to the degradedness condition, it is clear that $-I(U_2;Y_2|U_1)+I(U_2;Y_1|U_1)\leq 0$, which implies
\begin{equation}
  h(Y_2|U_1,U_2)-h(Y_1|U_1,U_2)\leq h(Y_2|U_1)-h(Y_1|U_1).
\end{equation}
Hence, there must exist an $\alpha_2$, $0\leq \alpha_2\leq1-\alpha_1$, such that
\begin{equation}\label{ieq5}
  h(Y_2|U_1,U_2)-h(Y_1|U_1,U_2)=\frac{1}{2}\log\frac{(1-\alpha_1-\alpha_2)P+N_2}{(1-\alpha_1-\alpha_2)P+N_1}.
\end{equation}
Substituting
\eqref{ieq6} and \eqref{ieq5} into \eqref{eq57}, we obtain
\begin{equation}
  R_2\leq \frac{1}{2}\log\left(1+\frac{\alpha_2P}{(1-\alpha_1-\alpha_2)P+N_2}\right)-\frac{1}{2}\log\left(1+\frac{\alpha_2P}{(1-\alpha_1-\alpha_2)P+N_1}\right).
\end{equation}
In order to bound $h(Y_1|U_1)$, we note that $Y_1$ can be written as $Y_2+Z'$, where $Z'$ is Gaussian with
variance $N_1-N_2$. By the entropy power inequality \cite{cover1991}, we have
\begin{equation}\label{eq63}
    2^{2h(Y_1|U_1)}\geq2^{2h(Y_2|U_1)}+2^{2h(Z'|U_1)}.
\end{equation}
Using \eqref{ieq6} and \eqref{eq63}, we obtain
\begin{equation}\label{eq64}
  H(Y_1|U_1)\geq \frac{1}{2}\log(2\pi e((1-\alpha_1)P+N_1)).
\end{equation}
Substituting \eqref{eq64} into \eqref{converse1}, we obtain
\begin{equation}
  R_1\leq \frac{1}{2}\log\left(1+\frac{\alpha_1P}{(1-\alpha_1)P+N_1}\right).
\end{equation}

We next bound $R_3$ and have
\begin{equation}\label{eq66}
\begin{split}
R_3\leq & h(Y_3|U_1,U_2)-h(Y_2|U_1,U_2)\\
&-(h(Y_3|U_1,U_2,U_3)-h(Y_2|U_1,U_2,U_3))+2\epsilon_n\\
\end{split}
\end{equation}
We let $\sqrt{t}=\sqrt{\frac{N_2-N_3}{N_1-N_3}}$, and write $Z_2=Z_3+\sqrt{t}\widetilde{Z}$, where $\widetilde{Z}$  is Gaussian  with variance $N_1-N_3$. By the entropy power inequality in \cite{Costa1985},
\begin{equation}
\begin{split}
 &2^{2h(Y_2|U_1,U_2)}=2^{2h(Y_3+\sqrt{t}\widetilde{Z}|U_1,U_2)} \geq(1-t)2^{2h(Y_3|U_1,U_2)}+t2^{2h(Y_1|U_1,U_2)},
\end{split}
\end{equation}
which, jointly with \eqref{ieq5}, implies that
\begin{equation}\label{eq68}
\begin{split}
 & h(Y_3|U_1,U_2)-h(Y_2|U_1,U_2)  \leq \frac{1}{2}\log\frac{(1-\alpha_1-\alpha_2)P+N_3}{(1-\alpha_1-\alpha_2)P+N_2}.
\end{split}
\end{equation}
Using the same argument for obtaining \eqref{ieq5}, we can show that there exists $0<\alpha_3<1-\alpha_1-\alpha_2$, such that
\begin{equation}\label{eq69}
  \begin{split}
    &h(Y_3|U_1,U_2,U_3)-h(Y_2|U_1,U_2,U_3)    =\frac{1}{2}\log\frac{(1-\alpha_1-\alpha_2-\alpha_3)P+N_3}{(1-\alpha_1-\alpha_2-\alpha_3)P+N_2}.
  \end{split}
\end{equation}
Substituting \eqref{eq68} and \eqref{eq69} into \eqref{eq66}, we obtain the desired bound on $R_3$.

It can be seen that the technique for bounding $R_3$ can be extended to bound $R_4,\ldots,R_K$, and we hence complete the proof by noting that $P_k=\alpha_kP$ for $k=1,\ldots,K$.
%
\section{Proof of Converse for Theorem \ref{MIMO}}\label{conv_gau_mimo}
In this proof, we first introduce some necessary definitions and useful lemmas in the previous studies \cite{ekrem2012,ekrem2011secrecy}. We then present our main proof.
\subsection{Preliminaries}
\begin{definition}\cite{ekrem2012}
Let $(\mathbf{U},\mathbf X)$ be an arbitrarily correlated length-$n$ random vector pair with well defined densities. The conditional Fisher information matrix of $\mathbf{X}$ given $\mathbf U$ is defined as
\begin{equation}
    \mathbf{J}(\mathbf{X}|\mathbf U)=E[\mathbf{\rho}(\mathbf{X|U})\mathbf\rho(\mathbf{X|U} )^T]
\end{equation}
where the expectation is taken over the joint density $f(\mathbf{u,x})$, and the conditional score function $\rho(\mathbf{x|u} )$ is given by
\begin{equation}
\begin{split}
    \rho(\mathbf{x|u} )&=\nabla \log f(\mathbf{x|u})\\
    &=\left[\frac{\partial\log f_U(\mathbf{x|u})}{\partial x_1}\cdots\frac{\partial \log f_U(\mathbf{x|u})}{\partial x_n}\right]^T.
\end{split}
\end{equation}
\end{definition}


\begin{lemma}\cite[Theorem 11]{ekrem2012}\label{thm11}
Let $(\mathbf{Z_1,Z_2,Z_3,Z_4})$ be Gaussian random vectors with covariance matrices $\mathbf{\Sigma}_1$, $\mathbf{\Sigma}_2$, $\mathbf{\Sigma}_3$, $\mathbf{\Sigma}_4$, respectively, where
\begin{equation}\label{thms11}
    \mathbf{\Sigma}_4\preceq \mathbf{\Sigma}_3 \preceq \mathbf{\Sigma}_2 \preceq \mathbf{\Sigma}_1.
\end{equation}
Let $(\mathbf{U},\mathbf{X})$ be an arbitrarily dependent random vector pair, which is independent of the Gaussian random vectors $(\mathbf{Z}_1,\mathbf{Z}_2,\mathbf{Z}_3,\mathbf{Z}_4)$, and the second moment of $\mathbf{X}$ be constrained as $E[\mathbf{XX}^T]\preceq\mathbf{S}$. Then, for any feasible $(\mathbf{U},\mathbf{X})$, for any $\mathbf{\Sigma_1,\Sigma_2,\Sigma_3,\Sigma_4}$ satisfying the order in \eqref{thms11}, there exists a positive semidefinite matrix $\mathbf{K}^*$ such that $\mathbf K^*\preceq \mathbf S$, and
\begin{equation}
h(\mathbf{X}+\mathbf{Z}_2|\mathbf{U})-h(\mathbf{X}+\mathbf{Z}_3|\mathbf{U})=\frac{1}{2}\log\frac{|\mathbf{K}^*+\mathbf{\Sigma}_2|}{|\mathbf{K}^*+\mathbf{\Sigma}_3|},
\end{equation}
and
\begin{flalign}
h(\mathbf{X}+\mathbf{Z}_1|\mathbf{U})-h(\mathbf{X}+\mathbf{Z}_3|\mathbf{U})\leq\frac{1}{2}\log\frac{|\mathbf{K}^*+\mathbf{\Sigma}_1|}{|\mathbf{K}^*+\mathbf{\Sigma}_3|},\\
h(\mathbf{X}+\mathbf{Z}_3|\mathbf{U})-h(\mathbf{X}+\mathbf{Z}_4|\mathbf{U})\geq\frac{1}{2}\log\frac{|\mathbf{K}^*+\mathbf{\Sigma}_3|}{|\mathbf{K}^*+\mathbf{\Sigma}_4|}.
\end{flalign}
\end{lemma}

\begin{lemma}\label{lemma2011}\cite{ekrem2011secrecy}
Let $(\mathbf{U,X})$ be an arbitrarily correlated random vector pair, and the second moment of $X$ is constrained as $E[\mathbf{XX}^T]\preceq S$. Let $\mathbf Z_1$, $\mathbf Z_2$ be Gaussian random vectors that are independent from $(\mathbf{U,X})$,  and have mean zero and covariance matrices $\mathbf\Sigma_1$, $\mathbf\Sigma_2$, respectively, where $\mathbf \Sigma_1\succeq\mathbf \Sigma_2$. Then
\begin{flalign}\label{eq:lemma3}
    \frac{1}{2}\log\frac{|\mathbf{J}(\mathbf X +\mathbf Z_1|\mathbf U)^{-1}|}{|\mathbf J(\mathbf X+\mathbf Z_1|\mathbf U)^{-1}+\mathbf\Sigma_2-\mathbf\Sigma_1|}&\leq h(\mathbf X+\mathbf Z_1|\mathbf U)-h(\mathbf X+\mathbf Z_2|\mathbf U)\nn\\
    &\leq \frac{1}{2}\log\frac{|\mathbf J(\mathbf X+\mathbf Z_2|\mathbf U)^{-1}+\mathbf\Sigma_1-\mathbf\Sigma_2|}{|\mathbf J(\mathbf X+\mathbf Z_2|\mathbf U)^{-1}|}.
\end{flalign}
\end{lemma}
\begin{proof}
The proof of the unconditioned version of Lemma \ref{lemma2011} is given in \cite{ekrem2011secrecy} in part B of Section V. The proof can be generalized to the conditioned version by applying mathematical tools given in part D of Section V of \cite{ekrem2011secrecy}.
\end{proof}
\begin{lemma}\cite[Lemma 17]{ekrem2011secrecy}\label{lemma4}
Let $(\mathbf{V,U,X})$ be $n-$dimentional random vectors with well-defined densities. Moreover, assume that the partial derivatives of $f(\mathbf u|\mathbf{v,x})$ with respect to $x_i$, $i=1,\ldots,n$, exist and satisfy
\begin{equation}
\underset{1\leq i\leq n}{max}\left|\frac{\partial f(\mathbf u|\mathbf{x,v})}{\partial x_i}\right|\leq g(\mathbf u),
\end{equation}
for some integrable function $g(\mathbf u)$. If $(\mathbf{V,U,X})$ satisfy the Markov chain $\mathbf V\rightarrow \mathbf U\rightarrow \mathbf X$, then
\begin{equation}
\mathbf J(\mathbf X|\mathbf U)\succeq \mathbf J(\mathbf X|\mathbf V).
\end{equation}
\end{lemma}
\begin{lemma}\cite[Lemma 10]{ekrem2011secrecy}\label{continuous}
Consider the function
\begin{equation}
r(t)=\frac{1}{2}\log\frac{|\mathbf A+\mathbf B+t\mathbf\Delta|}{|\mathbf A+t\mathbf\Delta|},\text{ }0\leq t\leq 1.
\end{equation}
where $\mathbf A$, $\mathbf B$, $\mathbf\Delta$ are real symmetric matrices, and $\mathbf A\succ\mathbf0$, $\mathbf B\succeq \mathbf 0$, $\mathbf\Delta\succeq\mathbf0$. Then $r(t)$ is continuous and monotonically decreasing with respect to $t$.
\end{lemma}

\begin{lemma}\label{lemaup}\cite{ekrem2011secrecy}
Suppose $(\mathbf U,\mathbf X)$ is a random vector pair with arbitrary joint distribution and the second order moment of $\mathbf X$ satisfies $E(\mathbf{XX}^T)\preceq S$. Let $\mathbf Z$ be a random Gaussian vector that is independent from $\mathbf U$ and $\mathbf X$ and has mean zero and covariance $\mathbf \Sigma$. Then we have
\begin{equation}
\begin{split}
0\preceq\mathbf J(\mathbf X+\mathbf Z|\mathbf U)^{-1}-\mathbf \Sigma\preceq \mathbf S.
\end{split}
\end{equation}
\end{lemma}
\begin{lemma}\cite{ekrem2011secrecy}\label{lemma8}
Suppose $(\mathbf U,\mathbf X)$ is a random vector pair with arbitrary joint distribution and the second order moment of $\mathbf X$ satisfies $E(\mathbf{XX}^T)\preceq S$. Let $\mathbf{Z_1,Z_2}$ be random Gaussian vectors that are independent from $\mathbf U$ and $\mathbf X$ and have mean zero and covariance matrices $\mathbf{\Sigma_1\succeq\Sigma_2}$.
Then we have
\begin{flalign}
\mathbf J(\mathbf X+\mathbf Z_1|\mathbf U)^{-1}+\mathbf \Sigma_2-\mathbf \Sigma_1-\mathbf J(\mathbf X+\mathbf Z_2|\mathbf U)^{-1}\succeq\mathbf 0.
\end{flalign}
\end{lemma}
The proof of Lemma \ref{lemma8} follows the arguments in the proof of Lemma 6 in \cite{ekrem2011secrecy} for the unconditional case, but using Corollary 4 in \cite{ekrem2011secrecy} for the conditional case.


\subsection{Main Proof}

Following the converse proof of Theorem \ref{th:capadmc} in Appendix \ref{conv_dmc}, we have the inequalities as follows:
\begin{equation}\label{bd1:dmc}
\begin{split}
R_1&\leq I(\mathbf U_1;\mathbf Y_1),\\
R_k&\leq I(\mathbf U_k;\mathbf Y_k|\mathbf U_{k-1})-I(\mathbf U_k;\mathbf Y_{k-1}|\mathbf U_{k-1}), \quad\text{for } 2\leq k\leq K,
\end{split}
\end{equation}
where the random variables satisfy the Markov chain condition in \eqref{Markov}.

We first derive the bounds on $R_2$ and $R_3$ in order to show that the bounding techniques can be extended to prove the bounds on $R_4,\ldots,R_K$. We then derive the bound on $R_1$.

To bound $R_2$, we start with \eqref{bd1:dmc}, and have
\begin{flalign}\label{bdr2}
        R_2&\leq I(\mathbf U_2;\mathbf Y_2|\mathbf U_1)-I(\mathbf U_2;\mathbf Y_1|\mathbf U_1)\nn\\
        &\overset{(a)}{=}h(\mathbf Y_2|\mathbf U_1)-h(\mathbf Y_2|\mathbf U_2)-(h(\mathbf Y_1|\mathbf U_1)-h(\mathbf Y_1|\mathbf U_2))\nn\\
        &=(h(\mathbf Y_2|\mathbf U_1)-h(\mathbf Y_1|\mathbf U_1))-(h(\mathbf Y_2|\mathbf U_2)-h(\mathbf Y_1|\mathbf U_2)),
\end{flalign}
where (a) follows from the Markov chain condition in \eqref{Markov}.

Following from Lemma \ref{lemma2011}, we have the following upper bound and lower bound on $h(\mathbf Y_2|\mathbf U_1)-h(\mathbf Y_1|\mathbf U_1)$.
\begin{flalign}\label{eq:bound}
\frac{1}{2}\log\frac{|\mathbf J(\mathbf X+\mathbf Z_2|\mathbf U_1)^{-1}|}{|\mathbf J(\mathbf X+\mathbf Z_2|\mathbf U_1)^{-1}+\mathbf\Sigma_1-\mathbf\Sigma_2|}&\leq h(\mathbf Y_2|\mathbf U_1)-h(\mathbf Y_1|\mathbf U_1)\nn\\
    &\leq \frac{1}{2}\log\frac{|\mathbf J(\mathbf X+\mathbf Z_1|\mathbf U_1)^{-1}+\mathbf\Sigma_2-\mathbf\Sigma_1|}{|\mathbf{J}(\mathbf X +\mathbf Z_1|\mathbf U_1)^{-1}|}.
\end{flalign}
Define
\begin{flalign}
\mathbf A&=\mathbf J(\mathbf X+\mathbf Z_2|\mathbf U_1)^{-1},\nn\\
\mathbf B&=\mathbf \Sigma_1-\mathbf \Sigma_2,\nn\\
\mathbf \Delta&=\mathbf J(\mathbf X+\mathbf Z_1|\mathbf U_1)^{-1}+\mathbf \Sigma_2-\mathbf \Sigma_1-\mathbf J(\mathbf X+\mathbf Z_2|\mathbf U_1)^{-1},\nn
\end{flalign}
and
\[r(t)=\frac{1}{2}\log\frac{|\mathbf A+\mathbf B+t\mathbf \Delta|}{|\mathbf A+t\mathbf \Delta|}.\]
Therefore, \eqref{eq:bound} can be rewritten into,
\begin{equation}\label{eq84}
-r(0)\leq h(\mathbf Y_2|\mathbf U_1)-h(\mathbf Y_1|\mathbf U_1)\leq -r(1).
\end{equation}
It can be verified that $\mathbf A \succ 0$, $\mathbf B\succeq0$, and $\mathbf \Delta\succeq0$.
In particular, $\mathbf \Delta\succeq0$ is due to Lemma \ref{lemma8}.

Following from Lemma \ref{continuous}, $r(t)$ is a continuous and monotonically decreasing function in $t$. Hence,  \eqref{eq84} implies that there must exist a constant $t_1$ with $0\leq t_1\leq 1$, such that
\begin{equation}
 h(\mathbf Y_2|\mathbf U_1)-h(\mathbf Y_1|\mathbf U_1)=-r(t_1).
\end{equation}
We define
\begin{flalign}
\mathbf S_1:&=\mathbf A+t_1\mathbf\Delta-\mathbf \Sigma_2\nn\\
&=\mathbf J(\mathbf X+\mathbf Z_2|U_1)^{-1}+t_1(\mathbf J(\mathbf X+\mathbf Z_1|\mathbf U_1)^{-1}+\mathbf \Sigma_2-\mathbf \Sigma_1-\mathbf J(\mathbf X+\mathbf Z_2|\mathbf U_1)^{-1})-\mathbf \Sigma_2.
\end{flalign}
Therefore,
\begin{equation}\label{r2part1}
h(\mathbf Y_2|\mathbf U_1)-h(\mathbf Y_1|\mathbf U_1)=-r(t_1)=\frac{1}{2}\log\frac{|\mathbf S_1+\mathbf \Sigma_2|}{|\mathbf S_1+\mathbf \Sigma_1|}.
\end{equation}
It can be seen that $\mathbf S_1$ satisfies $\mathbf A-\mathbf \Sigma_2\preceq \mathbf S_1\preceq \mathbf A+\mathbf \Delta-\mathbf \Sigma_2$. Following Lemma \ref{lemaup}, we have
\begin{equation}
    \mathbf 0\preceq \mathbf J(\mathbf X+\mathbf Z_2|\mathbf U_1)^{-1}-\mathbf \Sigma_2=\mathbf A-\mathbf \Sigma_2\preceq \mathbf S_1\preceq\mathbf A+\mathbf \Delta-\mathbf \Sigma_2=\mathbf J(\mathbf X+\mathbf Z_1|\mathbf U_1)^{-1}-\mathbf \Sigma_1\preceq\mathbf S,\\
\end{equation}
which implies
\begin{equation}
    \mathbf 0\preceq \mathbf S_1\preceq \mathbf S.
\end{equation}
We next study the term $h(\mathbf Y_2|\mathbf U_2)-h(\mathbf Y_1|\mathbf U_2)$. Due to the Markov chain condition \eqref{Markov}, it is clear that $-I(U_2;Y_2|U_1)+I(U_2;Y_1|U_1)\leq 0$, which implies that
\begin{equation}\label{part1}
    h(\mathbf Y_2|\mathbf U_2)-h(\mathbf Y_1|\mathbf U_2) \leq h(\mathbf Y_2|\mathbf U_1)-h(\mathbf Y_1|\mathbf U_1)=\frac{1}{2}\log\frac{|\mathbf S_1+\mathbf \Sigma_2|}{|\mathbf S_1+\mathbf \Sigma_1|}.
\end{equation}
Applying Lemma \ref{lemma2011}, we obtain
\begin{flalign}\label{part2}
\frac{1}{2}\log\frac{|\mathbf{J}(\mathbf X +\mathbf Z_2|\mathbf U_2)^{-1}|}{|\mathbf J(\mathbf X+\mathbf Z_2|\mathbf U_2)^{-1}+\mathbf\Sigma_1-\mathbf\Sigma_2|}&\leq h(\mathbf Y_2|\mathbf U_2)-h(Y_1|\mathbf U_2)\nn\\
    &\leq \frac{1}{2}\log\frac{|\mathbf J(\mathbf X+\mathbf Z_2|\mathbf U_2)^{-1}+\mathbf\Sigma_2-\mathbf\Sigma_1|}{|\mathbf J(\mathbf X+\mathbf Z_1|\mathbf U_2)^{-1}|}.
\end{flalign}
Combining \eqref{part1} and\eqref{part2}, we have
\begin{equation}\label{ulb}
    \frac{1}{2}\log\frac{|\mathbf{J}(\mathbf X +\mathbf Z_2|\mathbf U_2)^{-1}|}{|\mathbf J(\mathbf X+\mathbf Z_2|\mathbf U_2)^{-1}+\mathbf\Sigma_1-\mathbf\Sigma_2|}\leq h(\mathbf Y_2|\mathbf U_2)-h(Y_1|\mathbf U_2)\leq\frac{1}{2}\log\frac{|\mathbf S_1+\mathbf \Sigma_2|}{|\mathbf S_1+\mathbf \Sigma_1|}.
\end{equation}

We now consider the function
\begin{equation}
    r(t)=\frac{1}{2}\log\frac{|\mathbf A+\mathbf B+t\mathbf \Delta|}{|\mathbf A+t\mathbf \Delta|}
\end{equation}
with $\mathbf A$, $\mathbf B$ and $\mathbf \Delta$ being redefined as,
\begin{equation}
\begin{split}
    \mathbf A&=\mathbf J(\mathbf X+\mathbf Z_2|\mathbf U_2)^{-1}\\
    \mathbf B&=\mathbf \Sigma_1-\mathbf \Sigma_2\\
    \mathbf \Delta&=\mathbf S_1+\mathbf \Sigma_2-\mathbf J(\mathbf X+\mathbf Z_2|\mathbf U_2)^{-1},
\end{split}
\end{equation}
where $\mathbf A \succ 0$, $\mathbf B\succeq0$, and $\mathbf \Delta\succeq0$.
In order to show $\mathbf \Delta\succeq\mathbf 0$, we show that $\mathbf S_1\succeq \mathbf J(\mathbf X+\mathbf Z_2|\mathbf U_2)^{-1}-\mathbf \Sigma_2$. Using Lemma \ref{lemma4}, we have
\begin{equation}
    \mathbf J(\mathbf X+\mathbf Z_2|\mathbf U_2)\succeq \mathbf J(\mathbf X+\mathbf Z_2|\mathbf U_1).
\end{equation}
Hence,
\begin{equation}
    \mathbf J(\mathbf X+\mathbf Z_2|\mathbf U_1)^{-1}\succeq \mathbf J(\mathbf X+\mathbf Z_2|\mathbf U_2)^{-1}.
\end{equation}
Since $\mathbf S_1\succeq \mathbf J(\mathbf X+\mathbf Z_2|\mathbf U_1)^{-1}-\mathbf \Sigma_2$, we have
\begin{equation}\label{delta}
    \mathbf S_1\succeq \mathbf J(\mathbf X+\mathbf Z_2|\mathbf U_2)^{-1}-\mathbf \Sigma_2.
\end{equation}
Thus, \eqref{ulb} can be rewritten as
\begin{equation}
    -r(0)\leq h(\mathbf Y_2|\mathbf U_2)-h(\mathbf Y_1|\mathbf U_2)\leq -r(1).
\end{equation}
Since the function $r(t)$ is monotone and continuous, there exists a constant $t_2$ with $0\leq t_2\leq 1$ such that  $h(\mathbf Y_2|\mathbf U_2)-h(\mathbf Y_1|\mathbf U_2)=-r(t_2)$.
Let $\mathbf S_2=\mathbf A+t_2\mathbf \Delta-\mathbf \Sigma_2$, and obtain
\begin{equation}\label{r2part2}
    h(\mathbf Y_2|\mathbf U_2)-h(\mathbf Y_1|\mathbf U_2)=-r(t_2)=\frac{1}{2}\log\frac{|\mathbf S_2+\mathbf \Sigma_2|}{|\mathbf S_2+\mathbf \Sigma_1|}
\end{equation}
It can be seen that
\begin{equation}
    \mathbf J(\mathbf X+\mathbf Z_2|\mathbf U_2)^{-1}-\mathbf \Sigma_2=\mathbf A-\mathbf\Sigma_2\preceq\mathbf S_2\preceq\mathbf A+\mathbf \Delta-\mathbf\Sigma_2=\mathbf S_1.
\end{equation}
Therefore, combining \eqref{r2part1} and \eqref{r2part2}, we obtain
\begin{equation}
\begin{split}
    R_2&\leq \frac{1}{2}\log\frac{|\mathbf S_1+\mathbf \Sigma_2|}{|\mathbf S_1+\mathbf \Sigma_1|}-\frac{1}{2}\log\frac{|\mathbf S_2+\mathbf \Sigma_2|}{|\mathbf S_2+\mathbf \Sigma_1|}\\
    &=\frac{1}{2}\log\frac{|\mathbf S_1+\mathbf \Sigma_2|}{|\mathbf S_2+\mathbf \Sigma_2|}-\frac{1}{2}\log\frac{|\mathbf S_1+\mathbf \Sigma_1|}{|\mathbf S_2+\mathbf \Sigma_1|}.
\end{split}
\end{equation}

We next derive an upper bound on $R_3$, which is a necessary step to show that the proof techniques can be iteratively extended to bound $R_4,\ldots,R_K$. Following from \eqref{bd1:dmc}, we have
\begin{equation}\label{eq:r3}
    R_3\leq h(\mathbf Y_3|\mathbf U_2)-h(\mathbf Y_2|\mathbf U_2)-(h(\mathbf Y_3|\mathbf U_3)-h(\mathbf Y_2|\mathbf U_3)).
\end{equation}
Using Lemma \ref{thm11} and \eqref{r2part2}, we obtain
\begin{equation}\label{temp1}
    h(\mathbf Y_3|\mathbf U_2)-h(\mathbf Y_2|\mathbf U_2)\leq\frac{1}{2}\log\frac{|\mathbf S_2+\mathbf \Sigma_3|}{|\mathbf S_2+\mathbf \Sigma_2|}.
\end{equation}
Similarly to \eqref{part1}, due to the Markov chain condition \eqref{Markov}, we have
\begin{equation}\label{temp2}
    h(\mathbf Y_3|\mathbf U_3)-h(\mathbf Y_2|\mathbf U_3)\leq h(\mathbf Y_3|\mathbf U_2)-h(\mathbf Y_2|\mathbf U_2).
\end{equation}

Using Lemma \ref{lemma2011} and \eqref{temp1} and \eqref{temp2}, we have,
\begin{equation}
    \frac{1}{2}\log\frac{|\mathbf J(\mathbf X+\mathbf Z_3|\mathbf U_3)^{-1}|}{|\mathbf J(\mathbf X+\mathbf Z_3|\mathbf U_3)^{-1}+\mathbf \Sigma_2-\mathbf \Sigma_3|}\leq h(\mathbf Y_3|\mathbf U_3)-h(\mathbf Y_2|\mathbf U_3)\leq \frac{1}{2}\log\frac{|\mathbf S_2+\mathbf \Sigma_3|}{|\mathbf S_2+\mathbf \Sigma_2|}.
\end{equation}
It can be shown that $\mathbf S_2\succeq \mathbf J(\mathbf X+\mathbf Z_3|\mathbf U_3)^{-1}-\mathbf \Sigma_3$ by using Lemma \ref{lemma4} and Lemma \ref{lemma8}.
Then following the similar arguments that yield \eqref{r2part2}, we can show that there exists an $\mathbf S_3$, such that $\mathbf 0\preceq\mathbf S_3\preceq\mathbf S_2\preceq \mathbf S_1\preceq \mathbf S$ and
\begin{equation}\label{bd2r3}
    h(\mathbf Y_3|\mathbf U_3)-h(\mathbf Y_2|\mathbf U_3)=\frac{1}{2}\log\frac{|\mathbf S_3+\mathbf \Sigma_3|}{|\mathbf S_3+\mathbf \Sigma_2|}.
\end{equation}
Therefore, substituting \eqref{temp1} and \eqref{bd2r3} into \eqref{eq:r3}, we obtain
\begin{equation}
\begin{split}
    R_3&\leq \frac{1}{2}\log\frac{|\mathbf S_2+\mathbf \Sigma_3|}{|\mathbf S_2+\mathbf \Sigma_2|}-\frac{1}{2}\log\frac{|\mathbf S_3+\mathbf \Sigma_3|}{|\mathbf S_3+\mathbf \Sigma_2|}\\
    &=\frac{1}{2}\log\frac{|\mathbf S_2+\mathbf \Sigma_3|}{|\mathbf S_3+\mathbf \Sigma_3|}-\frac{1}{2}\log\frac{|\mathbf S_2+\mathbf \Sigma_2|}{|\mathbf S_3+\mathbf \Sigma_2|}.
\end{split}
\end{equation}
Using techniques similar to those for bounding $R_2$ and $R_3$, we can derive the desired bounds on $R_4,\ldots,R_K$ iteratively.

Finally, we bound the rate $R_1$.
We introduce a virtual receiver $\mathbf Y_0=\mathbf X+\mathbf Z_0$, where $\mathbf Z_0$ is a Gaussian vector with the covariance matrix of $\mathbf \Sigma_0=t\mathbf \Sigma_1$ with $t\geq 1$.
Hence, $\mathbf \Sigma_0\succeq\mathbf \Sigma_1$.
Following from \eqref{r2part1} and Lemma \ref{thm11},
%
we have,
\begin{equation}\label{ieq666}
    h(\mathbf Y_0|\mathbf U_1)-h(\mathbf Y_1|U_1)\leq \frac{1}{2}\log\frac{|\mathbf S_1+\mathbf \Sigma_0|}{|\mathbf S_1+\mathbf \Sigma_1|},
\end{equation}
for any $t\geq1$.
On the other hand, we have
\begin{equation}
    \frac{1}{2}\log(2\pi e)^r|\mathbf \Sigma_0|=h(\mathbf Z_0)\leq h(\mathbf Y_0|\mathbf U_1)\leq h(\mathbf Y_0)\leq \frac{1}{2}\log(2\pi e)^r|\mathbf S+\mathbf \Sigma_0|,
\end{equation}
which implies that
\begin{equation}
    \frac{1}{2}\log\frac{|\mathbf \Sigma_0|}{|\mathbf S_1+\mathbf \Sigma_0|}\leq h(\mathbf Y_0|\mathbf U_1)-\frac{1}{2}\log(2\pi e)^r|\mathbf S_1+\mathbf \Sigma_0|\leq \frac{1}{2}\log\frac{|\mathbf S+\mathbf \Sigma_0|}{|\mathbf S_1+\mathbf \Sigma_0|}.
\end{equation}
As $t\rightarrow \infty$,  $\frac{1}{2}\log\frac{|\mathbf \Sigma_0|}{|\mathbf S_1+\mathbf \Sigma_0|}\rightarrow 0$ and $\frac{1}{2}\log\frac{|\mathbf S+\mathbf \Sigma_0|}{|\mathbf S_1+\mathbf \Sigma_0|}\rightarrow0$. Hence, $h(\mathbf Y_0|\mathbf U_1)-\frac{1}{2}\log(2\pi e)^r|\mathbf S_1+\mathbf \Sigma_0|\rightarrow0$ as $t\rightarrow\infty$.
Since \eqref{ieq666} holds for any $t\geq1$, we have $h(\mathbf Y_1|\mathbf U_1)\geq \frac{1}{2}\log(2\pi e)^r|\mathbf S_1+\mathbf \Sigma_1|$.

Following from \eqref{bd1:dmc},
\begin{equation}
\begin{split}
R_1&\leq I(\mathbf U_1;\mathbf Y_1)\\
&=h(\mathbf Y_1)-h(\mathbf Y_1|\mathbf U_1)\\
&\leq \frac{1}{2}\log(2\pi e)^r|\mathbf S+\mathbf \Sigma_1|-\frac{1}{2}\log(2\pi e)^r|\mathbf S_1+\mathbf \Sigma_1|\\
&=\frac{1}{2}\log\frac{|\mathbf S+\mathbf \Sigma_1|}{|\mathbf S_1+\mathbf \Sigma_1|},
\end{split}
\end{equation}
which completes the proof.

\renewcommand{\baselinestretch}{1}
\begin{small}

\bibliographystyle{unsrt}
\bibliography{bibfile}

\end{small}

%

\end{document}